\documentclass[journal,10pt]{IEEEtran}
\usepackage{cite}
\usepackage{amsmath,amssymb,amsfonts}
\usepackage{algorithmic}
\usepackage{graphicx}
\usepackage{amsthm}
\usepackage{epsfig}
\usepackage{color}
\usepackage{subcaption}
\usepackage{lettrine}
\usepackage{float}
\usepackage{lipsum}
\usepackage{bm}
\usepackage{mathtools}
\usepackage{graphics}
\usepackage{bbm}
\usepackage{etoolbox}
\usepackage{suffix}
\usepackage[T1]{fontenc}
\usepackage[colorlinks]{hyperref} 
\usepackage[nameinlink,noabbrev]{cleveref}
\usepackage{textcomp}
\usepackage{xcolor}
\usepackage[def-file=diffcoeff-doc]{diffcoeff}
\usepackage{stackengine}
\usepackage[ruled,norelsize]{algorithm2e}

\usepackage{textcomp}
\def\BibTeX{{\rm B\kern-.05em{\sc i\kern-.025em b}\kern-.08em
    T\kern-.1667em\lower.7ex\hbox{E}\kern-.125emX}}

\usepackage{cleveref}

\newtheorem{lemma}{Lemma}

\newtheorem{remark}{Remark}

\newcommand{\E}{\ensuremath{\mathbb E}}
\newcommand{\q}{\ensuremath{\mathbf q}}
\renewcommand{\v}{\ensuremath{\mathbf v}}
\renewcommand{\u}{\ensuremath{\mathbf u}}
\newcommand{\s}{\ensuremath{\mathbf s}}
\renewcommand{\r}{\ensuremath{\mathbf r}}
\newcommand{\w}{\ensuremath{\mathbf w}}
\newcommand{\N}{\ensuremath{\mathcal N}}
\newcommand{\K}{\ensuremath{\mathcal K}}
\renewcommand{\P}{\ensuremath{\mathbf{ P}}}
\newcommand{\Q}{\ensuremath{\mathbf{Q}}}

\renewcommand{\O}{\ensuremath{\mathcal O}}

\def \treq {\stackrel{\tiny \Delta}{=}}

\makeatletter
\def\@seccntformat#1{\@ifundefined{#1@cntformat}%
	{\csname the#1\endcsname\quad}
	{\csname #1@cntformat\endcsname}
	}
\newcommand{\removelatexerror}{\let\@latex@error\@gobble}
\makeatother

\usepackage{scalerel}
\usepackage{tikz}
\usetikzlibrary{svg.path}

\definecolor{orcidlogocol}{HTML}{A6CE39}
\tikzset{
  orcidlogo/.pic={
    \fill[orcidlogocol] svg{M256,128c0,70.7-57.3,128-128,128C57.3,256,0,198.7,0,128C0,57.3,57.3,0,128,0C198.7,0,256,57.3,256,128z};
    \fill[white] svg{M86.3,186.2H70.9V79.1h15.4v48.4V186.2z}
                 svg{M108.9,79.1h41.6c39.6,0,57,28.3,57,53.6c0,27.5-21.5,53.6-56.8,53.6h-41.8V79.1z M124.3,172.4h24.5c34.9,0,42.9-26.5,42.9-39.7c0-21.5-13.7-39.7-43.7-39.7h-23.7V172.4z}
                 svg{M88.7,56.8c0,5.5-4.5,10.1-10.1,10.1c-5.6,0-10.1-4.6-10.1-10.1c0-5.6,4.5-10.1,10.1-10.1C84.2,46.7,88.7,51.3,88.7,56.8z};
  }
}

\newcommand\orcidicon[1]{\href{https://orcid.org/#1}{\mbox{\scalerel*{
\begin{tikzpicture}[yscale=-1,transform shape]
\pic{orcidlogo};
\end{tikzpicture}
}{|}}}}

\usepackage{hyperref} 

\IEEEoverridecommandlockouts
\begin{document}

\title{Terahertz Meets Untrusted UAV-Relaying:\\
Minimum Secrecy Energy Efficiency Maximization via Trajectory and Communication Co-design}

\author{\vspace{-1mm}Milad Tatar Mamaghani\textsuperscript{\orcidicon{0000-0002-3953-7230}}\,,~\IEEEmembership{Graduate Student Member,~IEEE}, Yi Hong\textsuperscript{\orcidicon{0000-0002-1284-891X}}\,,~\IEEEmembership{Senior Member,~IEEE}%
 \thanks{
  Copyright~(c)~2015 IEEE. Personal use of this material is permitted. However, permission to use this material for any other purposes must be obtained from the IEEE by sending a request to pubs-permissions@ieee.org. 
 \par Milad   Tatar   Mamaghani   and   Yi   Hong   are   with   the   Department of   Electrical   and Computer   Systems   Engineering,   Faculty of   Engineering,   Monash   University,   Melbourne,   VIC   3800,   Australia   (corresponding author e-mail:   milad.tatarmamaghani@monash.edu). 
  This research is supported by the Australian Research Council under Discovery Project DP210100412.
}
}
\maketitle
\begin{abstract}
Unmanned aerial vehicles (UAVs) and Terahertz (THz) technology are envisioned to play paramount roles in next-generation wireless communications.  In this paper, we present a novel secure UAV-assisted mobile relaying system operating at THz bands for data acquisition from multiple ground user equipments (UEs) towards a destination. We assume that the UAV-mounted relay may act, besides providing relaying services, as a potential eavesdropper called the untrusted UAV-relay (UUR). To safeguard end-to-end communications, we present a secure two-phase transmission strategy with cooperative jamming. Then, we devise an optimization framework in terms of a new measure $-$ secrecy energy efficiency (SEE), defined as the ratio of achievable average secrecy rate to average system power consumption, which enables us to obtain the best possible security level while taking UUR's inherent flight power limitation into account. For the sake of quality of service fairness amongst all the UEs, we aim to maximize the minimum SEE (MSEE) performance via the joint design of key system parameters, including UUR's trajectory and velocity, communication scheduling, and network power allocation. Since the formulated problem is a mixed-integer nonconvex optimization and computationally intractable, we decouple  it  into  four  subproblems and propose alternative algorithms to solve it efficiently  via greedy/sequential block successive convex approximation and non-linear fractional programming techniques. Numerical results demonstrate significant MSEE performance improvement of our designs compared to other known benchmarks.
\end{abstract}

\begin{IEEEkeywords}
UAV, THz, untrusted aerial relaying, physical layer security, minimum secrecy energy efficiency, trajectory design, resource allocation, convex optimization.
\end{IEEEkeywords}

\section{Introduction}
\lettrine[lines=2]{T}{he} unmanned aerial vehicle (UAV) has recently been recognized  as one of the major technological breakthroughs to be pervasively applied in 5G-and-beyond wireless communication networks, supporting massive machine-type communications, internet of things (IoT), and artificial intelligence (AI)-empowered communications \cite{jiang2021road, akyildiz20206g, geraci2021will}. Thanks to the unique characteristics of agility, on-demand swift deployment, versatility, and channel superiority amongst the other potentialities, UAV-aided wireless communications have recently attracted a great deal of research \cite{zeng2019accessing, Mozaffari2019Tutorial, wu2021comprehensive, li2018uav, zeng2016wireless, zhang2020multi}. Despite numerous advantages, the open nature of air-ground (AG) links inevitably makes such systems vulnerable to malicious attacks such as eavesdropping. Accordingly, the security and confidentiality of such promising wireless communication systems are of utmost concern and undeniable requirements. To protect the confidentiality of UAV communications against hostile entities, one promising technique is the physical layer security (PLS) that uses the characteristics of wireless channels and applies communication techniques to combat attacks without complex encryption. A number of works have found leveraging the PLS in UAV-aided communications plausibly effective  \cite{wang2019uav, sun2019physical, wang2020energy, tatarmamaghani2021joint, xu2020low, Hongliang2018Sec, Li2019Coo, mamaghani2020improving, Lee2018UAV, wu2021uav, mamaghani2019performance, Wang2018Joi, song2018joint, yuan2019joint, mamaghani2020intelligent}.  For example, PLS has been exploited in a wireless-powered UAV-relay system to combat eavesdropping via maximizing secrecy rate by a joint design of UAV's position and resource allocation \cite{wang2020energy}. Other efforts were made to maximize the average secrecy rate (ASR) via joint trajectory and communication design for UAV-standalone wireless system \cite{tatarmamaghani2021joint, xu2020low, Hongliang2018Sec}, for double-UAV with external jamming \cite{Li2019Coo, mamaghani2020improving, Lee2018UAV}, and for secure UAV-relaying scenarios \cite{wu2021uav, mamaghani2019performance, Wang2018Joi, song2018joint, yuan2019joint, mamaghani2020intelligent}. The majority of previous research has deemed the UAV to be a fully authorized and legitimate communication node in UAV-assisted relaying applications. However, when the UAV behaves as an \textit{untrusted relay}, which is called untrusted UAV-relay (UUR), with the capability of information eavesdropping while assisting end-to-end communications (see \cite{nuradha2019physical, Mamaghani2018sec}), the system design becomes quite challenging and entirely different from the existing body of research.

Further, energy efficiency is another imperative need for UAV-aided communications due to UAVs' inherent constraints on size, weight, and power (SWAP). Typically, the small-scale rotary-wing UAVs are powered via limited on-board batteries, leading to a restrictive operational lifetime, which undoubtedly impacts their overall system performance. Nonetheless, UAVs' flight endurance, if properly designed, can be enhanced to a considerable extent \cite{zeng2019energy}. Most recently, some works have studied the secrecy performance of UAV-aided systems considering the propulsion energy consumption constraint \cite{zhang2021dual, cai2020joint, xiao2019secrecy, miao2021secrecy}.  In \cite{zhang2021dual}, the authors have investigated ASR maximization for a cooperative dual-UAV secure data collection with propulsion energy limitation. Exploring the problem of secrecy energy efficiency (SEE) maximization for UAV-aided wireless systems is another research path \cite{cai2020joint, xiao2019secrecy, miao2021secrecy}.  The authors have designed both trajectory and resource allocation for the energy-efficient secure UAV communication system with the help of a multi-antenna UAV-jammer in \cite{cai2020joint}.  Some appropriate system designs have been conducted for the SEE improvement of a single UAV-relay system  \cite{xiao2019secrecy}, and a UAV-swarm multi-hop relaying scenario \cite{miao2021secrecy}. It is worth pointing out that all the aforementioned designs have only aimed to combat external terrestrial eavesdroppers.

On the other hand, owing to the ultra-broad bandwidth at the terahertz (THz) frequency range ($0.1-10~\mathrm{THz}$), THz transmission has been acknowledged as a promising technology capable of catering an explosive growth of user demand of higher mobile traffic for future wireless systems \cite{chen2019survey}. However, THz links incur severe path loss and high susceptibility to environmental blockage, and molecular absorption \cite{jiang2021road, chaccour2021seven, wu2021comprehensive}, which limit signal propagation distance and coverage range. To overcome the hindrances, one possible solution might be exploiting UAV-aided communications in THz links. Notably, in the context of THz-UAV systems, few initial research studies have thus far been conducted. The coverage probability of the UAV-THz downlink communications was analyzed in \cite{wang2020performance}, while \cite{xu2021joint} has explored a similar non-security scenario with a focus on minimizing communication delay by the joint design of a UAV's location and power control. When it comes to security issues of such high-frequency systems, despite the widely-assumed improved resiliency against eavesdropping of THz links, the authors of \cite{ma2018security} have characterized the possibility of eavesdropping attacks for such systems. Needless to mention that even considering negligible information leakage towards the external malicious eavesdroppers through THz transmissions, the scenarios involving untrusted relays, particularly the UUR systems, may still be vulnerable to eavesdropping. The appropriate design for such systems has yet to be understood; therefore, one needs to design novel frameworks to enable the efficient deployment of THz-UUR wireless systems.
 
\subsection{Our contributions}
To the best of our knowledge, this is the first work addressing the energy-efficient secure design of a THz-UUR wireless communication system to guarantee confidentiality and transmission secrecy with the least system power consumption. Our detailed contributions are summarized below.
\begin{itemize}
    \item 
    We present an UUR-enabled wireless communication system for data collection from multiple ground user equipments (UEs) towards the base station (BS) over THz-based AG links. We adopt a secure two-phase transmission strategy using destination-assisted cooperative jamming (DACJ) to improve security. 
    
    \item Then, we formulate a maximization problem in terms of a new measure {\em minimum secrecy energy efficiency} (MSEE), defined as the minimum ratio of achievable ASR to average system power consumption. This optimization problem leads to a joint design of key system parameters, including UUR's trajectory and velocity, communication scheduling, and network transmission power allocations.
    
    \item
    Since the optimization  problem is originally intractable due to non-convexity, we decompose it into four subproblems and then solve each inspired by the successive convex approximation (SCA) or Dinkelbach fractional programming techniques. Further, we propose two computationally efficient algorithms according to the sequential and maximum improvement (MI) based block coordinate descent (BCD) approaches with guaranteed convergence to at least a suboptimal solution. We also thoroughly conduct computational and complexity analysis and show that our solution can be obtained in polynomial time order, making it applicable to the energy-hungry UAV-based scenarios.

    \item
    We conduct extensive simulations to verify the analyses and demonstrate the effectiveness of our proposed designs in terms of MSEE compared to some other benchmarks; i.e., without communication resource allocation design or trajectory and velocity optimization and ignoring flight power consumption. We also investigate the impact of some fundamental setting parameters such as the flight mission time, the molecular adsorption factor, the average network transmit power budget, and the flight power limit  on the overall system performance. 

\end{itemize}

\textit{Notation}: We use bold lower case letters to denote vectors.
$\|\cdot\|$ denotes Frobenius norm; $\E\{x\}$ stands for
the expectation over the random variable (r.v.) $x$; $\mathcal{CN}(\mu, \sigma^2)$ denotes a circularly symmetric complex Gaussian r.v. with mean $\mu$ and variance $\sigma^2$;  $(x)^+ = \max\{x,0\}$ and $\max$ stands for the maximum value; $\lfloor x \rfloor$ is the smallest integer that is larger than or equal to $x$; $\mathcal{O}(\cdot)$ denotes the big-O notation.

The rest of the paper is organized as follows. Section~\ref{sec:sysmodel} introduces system model and formulates the problem of interest. In Section~\ref{sec:solution}, we present efficient iterative algorithms to solve the optimization problem, followed by numerical results and discussions given in Section \ref{sec:numerical}. Finally, we draw our conclusions in Section~\ref{sec:conclusion}.


\section{System model and problem formulation} \label{sec:sysmodel}
\begin{figure}[t]
\centerline{\includegraphics[width= 0.8\columnwidth]{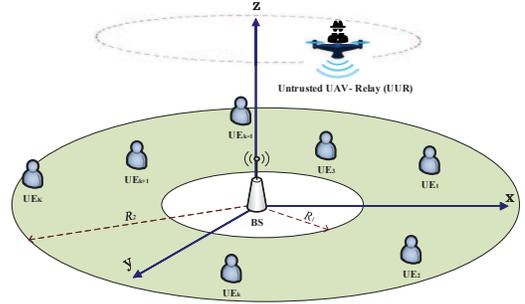}}
\caption{\footnotesize System model of secure untrusted mobile UAV-relaying via THz communications with cooperative jamming.}
\label{fig1_sysmodel}
\end{figure}

We consider a UAV-enabled wireless communication system for data collection from a set of $K$ ground UEs towards a BS via a UAV-assisted mobile amplify-and-forward (AF) relay, as shown in Fig. \ref{fig1_sysmodel}. Here we assume that there are no reliable direct links from UEs to BS (see \cite{Mamaghani2018sec, mamaghani2017secure} and references therein), and all nodes are equipped with a {\em single antenna}, operating in {\em half-duplex} mode. Therefore, a UAV-relay is employed to assist end-to-end communications \cite{mamaghani2019performance}; nonetheless, the UAV-relay may not be fully authorized to access collected confidential information and may conduct malicious eavesdropping, i.e., a UUR \cite{nuradha2019physical}. Thus, secure data transmission is in demand.

\begin{remark}
Note that in this work, we consider a single-antenna untrusted UAV-relaying scenario focusing on the secure energy-efficient design. Extension to energy-constraint multiple-antenna UAV \cite{zhang2020multi} under THz bands is intriguing but challenging, which can be investigated in future works, particularly taking into account the following factors. Firstly, the multi-input multi-output (MIMO) signal processing 
not only brings on higher complexity with the hardware costs due to several active elements and radio frequency (RF) chain, but also more energy consumption, which is constrained by the UAV's SWAP limitations, and should be carefully considered in the future designs. Further, the lack of rich scattering in the UAV environment, compared to the terrestrial communication systems,  considerably limits the spatial multiplexing gain of MIMO, leading to only marginal rate improvement over single-antenna UAV systems \cite{zeng2016wireless}. Last but not least, due to UAV's highly dynamic environment, it would be particularly challenging to achieve transceiver beam-alignment for directional beamforming, while otherwise, beam misalignment degrades the multi-antenna diversity gain.
\end{remark}



Without loss of generality, we consider a three-dimensional (3D) Cartesian coordinate system, where the BS's horizontal coordinate is located at the origin $\q_b=[0,0] \in \mathbb{R}^{1\times2}$, and the ground UEs with horizontal fixed\footnote{
Note that in a more realistic environment all the network nodes, including terrestrial UEs, apart from the UUR, might be mobile depending on the application \cite{khuwaja2019effect}. However, considering the high velocity and flexibility of the UAV compared to the ground nodes, we consider the ground terminals quasi-stationary throughout the flight mission with unchanged coordinates.} coordinates $\q_{k}=[x_k, y_k] \in \mathbb{R}^{1\times2}$ for $\forall k\in\K$, where $\K=\{1, 2,\cdots, K\}$, are randomly distributed in a circular annulus region with the inner radius $R_1$ and outer radius $R_2$ and the coordinates are assumed to be known in prior. Here, $R_1$ is considered to be the largest distance at which a reliable uplink transmission can be obtained, while beyond $R_1$ in our case implies no direct link between UE and BS. Further, $R_2$ indicates the boundary of {\em the permitted flying region} for the UAV to provide communication service. 

We also consider that UAV flies from and back to the same specific point over the region of interest for a duration of $T$ seconds in order to provide relaying services to all UEs with fairness. This specific point may refer to the checkup point wherein the UAV gets recharged and physically examined to maintain its service. Assuming that UAV flies at a fixed altitude\footnote{Fixed-altitude operation can be justified  from a practical viewpoint in order for UAV to cut off extra energy consumption arising from ascending or descending \cite{Hongliang2018Sec}, \cite{Li2019Coo}.} $H$ meters whose instantaneous horizontal coordinate and velocity is represented by $\q(t)=[x(t), y(t)]$ and $\v(t)\treq\diff[1]{\q(t)}{t}$, respectively, where $0 \leq t \leq T$. For the ease of analysis, we adopt the time-slotted system  such that the flight duration $T$ is equally discretized into $N$ sufficiently small time slots of duration $\delta_t \treq \frac{T}{N}$. Hence, the UAV's horizontal location at time slot $n\in \N =\{1,\cdots,N\}$ can be approximated by $\q[n]=[x[n], y[n]]$. This assumption is valid when $d^{max}_{t} \treq \delta_t v^{max}_u \ll H$, wherein $d^{max}_{t}$ denotes the maximum UAV's displacement per time slot.


\subsection{Channel model}
We assume that the AG links are over THz channels, which are mainly affected by both free space spreading loss and the molecular absorption according to \cite{xu2021joint}. Further, as per widely used assumption in the body of existing UAV literature, e.g., \cite{zeng2017energy,cai2020joint,tatarmamaghani2021joint}, the Doppler effect due to the UAV mobility is also considered to be perfectly compensated for ease of exposition in the sequel. Therefore, assuming that at each time slot $n$ the channel state information is regarded static due to adequately small $\delta_t$, we adopt the line-of-sight (LoS) dominant time-varying THz channel power gain model, similar to \cite{xu2021joint}, between the UUR and any UE $k \in \K$ as 
\begin{align}\label{huk}
    h_{ku}[n] = \frac{\beta_0 \exp(-a_f d_{ku}[n])}{d^2_{ku}[n]},~\forall n
\end{align}
where $d_{ku}[n]$ denotes the Euclidean distance between the UUR and the $k$-th UE, given by
\begin{align}
    d_{ku}[n] = \sqrt{\|\q[n]-\q_k\|^2+H^2},~\forall n
\end{align}
Note that the multiplicative term $\exp(-a_f d_{ku})$ in \eqref{huk} is the indication of excessive path loss of THz links due to water vapor molecular absorption effect\footnote{It is worth stressing that THz transmissions generally admit various peculiarities such as molecular absorption effect, spectral broadening, frequency selectivity, and so forth \cite{tekbiyik2020holistic, chaccour2021seven}. In light of this, to confront the high frequency-selectivity nature in the THz band, the total bandwidth of the THz frequencies is generally divided into several sub-bands \cite{pan2021uav}. Therefore, this work considers only one sub-band equally shared amongst communication nodes with the associated carrier frequency, and the molecular absorption effect is the solely peculiar trait we take into account in this work for ease of exposition as in \cite{xu2021joint, pan2021uav}.}, wherein $a_f$ is the frequency and chemical composition of air dependent adsorption factor \cite{boulogeorgos2018distance}. It should be also pointed out that the main cause of absorption loss in THz frequency ranges is the water vapor molecules that cause discrete, but deterministic loss to the signals in the frequency domain.
Further, $\beta_0 \treq (\frac{C}{4\pi f})^2$ denotes the reference channel power gain at unit distance, wherein $C$ is the speed of light, $f$ is the operating frequency. Likewise, the THz channel power gain between the UUR and the BS can be written as 
$h_{bu}[n] = \frac{\beta_0 \exp(-a_f d_{bu}[n])}{d^2_{bu}[n]}$, where $d_{bu}[n] = \sqrt{\|\q[n]-\q_b\|^2+H^2_u},~\forall n$.
\subsection{Constraints on user scheduling, power, UAV's mobility}
We adopt the time division multiple access (TDMA) protocol for multi-user relaying services, wherein UUR serves at most one enabled UE at $n$-th time slot, while the other ground UEs keep silent. Therefore, letting $\zeta_k[n]$ be a binary user scheduling variable for UE $k\in \K$ at time slot $n\in \N$, we have the user scheduling constraints as
\begin{align}
&\mathrm{C1}:\quad    \zeta_k[n] \in \{0, 1\}, \quad\forall k, n\\
&\mathrm{C2}:\quad    \sum_{k\in\K} \zeta_k[n] \leq 1,\quad\forall n
\end{align}
where $\zeta_k[n]=1$ if UE $k$ is scheduled at time slot $n$, and $\zeta_k[n]=0$, otherwise.
Further, the transmit powers of the UUR, the BS, and $k$-th user in time slot $n$, denoted respectively as $p_u[n], p_b[n]$, and $p_k[n]$, are generally subject to average and peak transmit powers given as
\begin{align}
&\mathrm{C3}:\frac{1}{N}\sum_{n=1}^{N}\sum_{k=1}^{K}\zeta_k[n]p_k[n] \leq p^{ave}_k,\\
&\mathrm{C4}:~0 \leq p_k[n] \leq p^{max}_k, \quad\forall k, n\\
&\mathrm{C5}:\frac{1}{N}\sum_{n=1}^{N}p_u[n] \leq p^{ave}_u,\\
&\mathrm{C6}:~0 \leq p_u[n] \leq p^{max}_u,\quad\forall n \\
&\mathrm{C7}:\frac{1}{N}\sum_{n=1}^{N}p_b[n] \leq p^{ave}_b,\\
&\mathrm{C8}:~0 \leq p_b[n] \leq p^{max}_b, \quad\forall n
\end{align}
where sets $\{p^{ave}_u, p^{ave}_b, p^{ave}_k, \forall k \}$ and $\{p^{max}_u, p^{max}_b, p^{max}_k, \forall k \}$ represent the corresponding  average and maximum  power constraints of the network nodes.

The mechanical power consumption of the energy-limited UAV due to high demand of propulsion energy for aerial operation with fixed-altitude flight can be approximately given by \cite{zeng2019energy}
\begin{align}\label{flightpow}
    P_{f}[n] &= \stackrel{}{\underset{\text{blade profile}}{\underbrace{{P_0\left(1+\frac{3\|\v[n]\|^2}{\Omega^2_uR_u^2}\right)}}}} +\stackrel{}{\underset{\text{induced}}{\underbrace{{\frac{1}{2}d_0\rho sA_u \|\v[n]\|^3}}}} \nonumber\\
    &+ \stackrel{}{\underset{\text{parasite}}{\underbrace{{P_i\left(\sqrt{1+\frac{\|\v[n]\|^4}{4\nu^4_0}} - \frac{\|\v[n]\|^2}{2\nu^2_0}\right)^{\frac{1}{2}}}}}}, \quad\forall n
\end{align}
wherein $\v[n]$ is the UAV's instantaneous velocity at time slot $n$, $P_0$ and $P_i$ are two constants representing UAV's \textit{blade profile power} and \textit{induced power} in hovering mode, respectively, $\Omega_u$ and $R_u$ are the UAV's blade angular velocity in Radian per second (rad/s) and its rotor radius in meter (m), $d_0, \rho, s$, and $A$ indicate the unit-less fuselage drag ratio, air density in $\mathrm{kg}/\mathrm{m}^3$, rotor solidity, and rotor disk area in $\mathrm{m}^2$, respectively. Further, the average rotor induced velocity in hovering is shown as  $\nu_0$. Thus, we have the average flight power consumption constraint as
\begin{align}
\mathrm{C9}:\frac{1}{N}\sum_{n=1}^{N}P_{f}[n] \leq \bar{P}_{lim},
\end{align}
wherein $\bar{P}_{lim}$ indicates the UAV's average propulsion power budget, which is proportional to the UAV's on-board battery capacity. Therefore, it should be required that the \textit{total consumed propulsion energy} by the UAV over $N$ time slots be less than such limit in order for network functioning. Further, the considered scenario should be subject to UAV's mobility constraints in terms of initial and final locations for cyclic path,  in-flight maximum displacement per time slot for satisfying channel invariant assumption, and permitted flying region as 
\begin{align}
&\mathrm{C10}:\quad    \q[0] = \q[N] = \q_I,\nonumber\\
&\mathrm{C11}:\quad    \q[n+1] = \q[n] +\v[n]\delta_t,\quad\forall n \setminus N\nonumber\\
&\mathrm{C12}:\quad\|\v[n]\|\leq v^{max}_u,\quad\forall n\nonumber\\
&\mathrm{C13}:\quad    \|\v[n+1] - \v[n]\| \leq a^{max}_u,\quad\forall n \setminus N\nonumber\\
&\mathrm{C14}:\quad    \|\q[n] - \q_b\| \leq R_2,\quad\forall n
\end{align}
wherein $\q_I$ indicates UAV's initial and final location per flight, $v^{max}_u$ and $a^{max}_u$ are the UAV's maximum speed and acceleration, respectively.

\subsection{Secure transmission strategy, problem formulation}
For the purpose of wireless security, we adopt a {\em secure two-phase transmission strategy with destination-assisted cooperative jamming} (DACJ) technique similar to \cite{mamaghani2019performance, Mamaghani2018sec, tatarmamaghani2021joint}. In the first phase, at each time slot $n$, the scheduled UE $k$ sends confidential information to UUR, and simultaneously the BS jams UUR.
As such, the received signal at UUR in time slot $n$ can be represented by
\begin{align}
    \hspace{-4mm}x_u[n] \hspace{-1mm}=\hspace{-1mm} \sqrt{p_k[n]} h_{ku}[n] s_k[n] \hspace{-1mm}+\hspace{-1mm} \sqrt{p_b[n]} h_{bu}[n] s_b[n] \hspace{-1mm}+ \hspace{-1mm}n_u[n],~\forall n
\end{align}
where $s_k[n]$ is the normalized information signal, i.e. $\E\{\|s_k[n]\|^2\}=1, \forall k, n$. Plus, $s_b[n] \sim \mathcal{CN}(0,1)$ represents BS's jamming transmission with unit power, and $n_u[n] \sim  \mathcal{CN}(0,\sigma^2_{u})$ denotes the additive white Gaussian noise (AWGN) at UUR. In the second phase, UUR forwards an amplified version of the received signals using AF relaying protocol to the BS. Here we assume that the amplification process is based on the full channel state information of UE-UUR links, i.e., the exact locations of UEs are known to the UUR, thus the normalized variable-gain of AF relaying from UE $k$ can be obtained as
\begin{align}
    \hspace{-3mm}G_k[n] \hspace{-1mm}=\hspace{-1mm} \sqrt{\frac{1}{p_k[n] \|h_{ku}[n]\|^2 \hspace{-1mm}+\hspace{-1mm} p_b[n] \|h_{bu}[n]\|^2 \hspace{-1mm}+\hspace{-1mm} N_0}},~\forall k, n
\end{align}
where $N_0$ indicates the noise power at UUR.
Accordingly, the resultant signal at the BS, after self-interference cancellation \cite{Mamaghani2018sec}, can be expressed as
\begin{align}
    y_b[n] &= \sqrt{p_k[n]p_u[n]}h_{ku}[n]h_{bu}[n] G_k[n] s_k[n]\nonumber\\
    &+ \sqrt{p_u[n]}G_k[n]h_{bu}[n]n_u[n] + n_b[n],~\forall n
\end{align}
where $n_b[n] \sim  \mathcal{CN}(0,\sigma^2_b)$ indicates the AWGN at the BS.
Under such setting, given the equally shared communication bandwidth $B$ Hz, the achievable end-to-end instantaneous data rate in bits-per-second (bps) from the $k$-th UE towards the BS at time slot $n$ is given by
\begin{align}
    R^{k}_{b}[n] &\hspace{-1mm}=\hspace{-1mm} \zeta_k\hspace{-0.5mm}[n]B\log_2\hspace{-1mm}\left(\hspace{-1mm}1\hspace{-1mm}+\hspace{-1mm}\frac{p_k[n] g_{ku}[n] p_u[n]g_{bu}[n]}{\left(p_u[n]\hspace{-1mm}+\hspace{-1mm}p_{b}[n]\right)g_{bu}[n] \hspace{-1mm}+\hspace{-1mm} p_k[n] g_{ku}[n] \hspace{-1mm}+\hspace{-1mm} 1}\hspace{-1mm}\right)\hspace{-0.5mm},
\end{align}
wherein $g_{ku}[n] \treq \frac{h_{ku}[n]}{N_0}$ and $g_{bu}[n] \treq \frac{h_{bu}[n]}{N_0}$, and $N_0\treq B\sigma^2_{u(b)}$. Then the UUR may overhear the confidential information with an achievable wiretap secrecy rate per Hz at time slot $n$ as
\begin{align}\label{ave_rsec}
    R^{k}_{u}[n] =  \zeta_k[n]B\log_2\left(1+\frac{p_k[n] g_{ku}[n]}{p_b[n]g_{bu}[n] + 1}\right),
\end{align}
 $N_0 \treq B\sigma^2_{u(b)}$ indicates the equal noise power at the receivers, which is assumed for simplicity of exposition. 

We adopt the ASR as one of the key secrecy metrics and the ASR of $k$-th UE at time slot $N$ is
\begin{align}\label{asr_k}
    \bar{R}^{k}_{sec} = \frac{1}{N}\sum^{N}_{n=1}\left[\frac{1}{2}\left(R^{k}_{b}[n]  - R^{k}_{u}[n]\right)^+\right] ~\text{bps}
\end{align}
wherein $(x)^+\treq\max\{x,0\}$, and the ratio $\frac{1}{2}$ is due to the fact that secure transmission is done in two phases of equal duration at each time slot. The achievable average  information bits can securely be exchanged between $k$-th UE and BS is 
$B^{k}_{sec} = \delta_t \sum^{N}_{n=1}{R}^{k}_{sec}[n].$ 

To fully exploit the capability of aerial platforms for communication, the limited energy resource must be considered in system design. Therefore, the total energy consumption of the UAV generally consists both the communication-related energy as well as mechanical-related propulsion energy to support the UAV's hovering and mobility. In practice, the UAV's propulsion power consumption is much higher than those used for communication purposes such as UEs' signal transmission, BS's jamming, and signal processing as well as other circuitry. Hence, we approximate the network's total power consumption by the amount that the UAV consumes for the propulsion purpose. Consequently, for the secrecy metric, we define {\em secrecy energy efficiency} (SEE) of the proposed scheme for the $k$-th UE as the {\em ratio of the achievable ASR to the approximated average system power consumption} as
\begin{align}\label{see_k}
    \mathbf{SEE}^{k}(\pmb{\zeta}, \Q, \P) \treq \frac{\bar{R}^{k}_{sec}}{\frac{1}{N}\sum^{N}_{n=1}P_{f}[n]},~\text{bits/Joule}
\end{align}
wherein the user scheduling set $\pmb{\zeta}= \{\zeta^k[n], \forall n, k\}$, UAV's location and velocity set $\mathbf{Q} = \{\q[n], \v[n],  \forall n \}$, and  network transmit power set $\P= \{\P_k=\{p_k[n], \forall k, n\}, \P_u=\{p_u[n], \forall n\}, \P_b=\{p_b[n], \forall n\}\}$ are the involving parameters. 
\begin{remark}
It is worth pointing out that for later analysis we utilize the normalized metrics, i.e., the numerator and denominator of \eqref{see_k} are divided by $B$ and $\bar{P}_{lim}$, respectively, to well balance numerical values of both metrics in SEE.
\end{remark}
To design the network so as to obtain the best performance among UEs and provide a fair quality of service (QoS) to all UEs given UAV's stringent on-board battery, we maximize the minimum SEE (MSEE) performance of the system by 
\begin{align}\label{opt_prob}
(\mathrm{P}):& \stackrel{}{\underset{\pmb{\zeta}, \mathbf{Q}, \mathbf{P}}{\mathrm{max}}~~\min_{k\in\K}  \mathbf{SEE}^{k}(\pmb{\zeta}, \mathbf{Q}, \mathbf{P}) } \nonumber\\
&~~\text{s.t.}~~~~~ \mathrm{C1-C14,}
\end{align}
We note that the problem $\mathrm{(P)}$ is a mixed-integer non-convex optimization problem, which is too hard to solve optimally. The non-convexity is mainly due to the non-concave objective function with respect to (w.r.t) the optimization variables, and also having the non-smoothness operator $(\cdot)^+$ and the non-convex constraints $\mathrm{(C1)}$, $\mathrm{(C3)}$, and $\mathrm{(C9)}$. Indeed, the major challenge in solving $(\mathrm{P})$ arises from the binary user scheduling constraint $\mathrm{C1}$ and the highly coupled optimization variables in the objective function in fractional form. To make it tractable, we first remove the operator $(\cdot)^+$ from the numerator of the objective function, since the value of the objective function should be non-negative at the optimal point; otherwise, one can set, e.g., $\mathbf{P}_k=\mathbf{0}, \forall k$ and get zero MSEE performance without modifying the original problem. Nonetheless, having at least a differentiable objective function, the problem is still non-convex, thereby no standard approach to solve it efficiently.  To remedy this issue, we first handle the binary constraint as per the approach in \cite{wu2018joint}, by relaxing $\mathrm{C1}$  into continuous constraint. Then, we propose some computationally efficient algorithms to iteratively solve a sequence of approximated convex subproblems by adopting several techniques such as block coordinated descent (BCD), successive convex approximation (SCA), and nonlinear fractional Dinkelbach programming, discussed below.


\section{Proposed Iterative Solution}\label{sec:solution}
In this section, we split the problem $\mathrm{(P)}$ into four subproblems with different blocks of variables, then solve each block by block, while keeping the other blocks unchanged. Specifically, we delve into solving the joint user scheduling and transmit power optimization subproblem to optimize $(\pmb{\zeta}, \P_k)$, relaying and jamming power optimization subproblems to improve $\P_u$ and $\P_b$, and lastly, the joint trajectory and velocity optimization subproblem to optimize $\mathbf{Q}$. Then, the overall algorithms to iteratively attain the approximate solution of \eqref{opt_prob} will be given.
\subsection{Joint user scheduling and transmit power optimization}

First, we relax binary variables $\pmb{\zeta}$ into continuous real-valued set $\tilde{\pmb{\zeta}}=\{\tilde{\zeta_k}[n], \forall k, n\}$.
The relaxed version of $\pmb{\zeta}$ serves, indeed, as a time sharing factor for $k$-th UE at time slot $n$. Such a relaxation in general leads the objective value of the relaxed problem to be asymptotically tight upper-bounded by that of the original binary-constrained problem \cite{wu2018joint}. Next,  we define the auxiliary variables $\tilde{\P}_k= \{\tilde{p}_k[n] \treq p_k[n] \tilde{\zeta_k}[n], \forall k, n\}$. By introducing a slack variable $\psi$, the corresponding subproblem can be equivalently represented as
\begin{subequations}\label{usrpowsch_subprob}
\begin{align}
(\mathrm{P1}):& \stackrel{}{\underset{\psi, \tilde{\pmb{\zeta}}, \tilde{\P}_k}{\mathrm{maximize}}~~\psi} \nonumber\\
\text{s.t.}~~~\begin{split} \sum^{N}_{n=1}&\stackrel{}{\underset{\text{Term I}}{\underbrace{{\tilde{\zeta}_k[n]
\ln\left(1+\frac{C_{n}\tilde{p}_k[n]}{\tilde{p}_k[n]+D_{k, n}\tilde{\zeta}_k[n]}\right)}}}}\\
&-\stackrel{}{\underset{\text{Term II}}{\underbrace{{{\tilde{\zeta}_k[n]\ln\left(1+B_{k, n}\frac{\tilde{p}_k[n]}{\tilde{\zeta}_k[n]}\right)}}}}}  \geq \frac{\psi}{\lambda_1},\quad\forall k  
 \end{split}  \label{cvx_p1}\\
&\frac{1}{N}\sum_{n=1}^{N}\sum_{k=1}^{K}\tilde{p}_k[n] \leq p^{ave}_k,\label{19b}\\
&0 \leq \tilde{p}_k[n] \leq \tilde{\zeta_k}[n]p^{max}_k,\quad\forall k, n \\
& \sum_{k\in\K}  \tilde{\zeta_k}[n]\leq 1,\quad\forall n\\
& 0 \leq \tilde{\zeta_k}[n] \leq 1,\quad\forall k, n \label{19e}\\
& \sum^{N}_{n=1}\sum^{K}_{k=1}\left(\tilde{\zeta}_k[n] - \tilde{\zeta}^2_k[n]\right) \leq  0, \label{P1_binary}
\end{align}
\end{subequations}
where
\begin{alignat*}{2}
{\lambda_1} &= \frac{B}{2\ln 2 \sum^{N}_{n=1}P_{f}[n]} &B_{k, n} = \frac{g_{ku}[n]}{p_b[n]g_{bu}[n]+1} \nonumber\\
C_{n} &= g_{bu}[n] p_u[n] &D_{k, n} = \frac{g_{bu}(p_u[n]+p_b[n])+1}{g_{ku}[n]}.
\end{alignat*}
Note that the constraint \eqref{cvx_p1} should be met with equality at the optimal point; otherwise, the value of the objective function in problem $(\mathrm{P1})$ can still be increased by increasing $\psi$, which violates the optimality. Furthermore, jointly considering the constraints \eqref{19e} and \eqref{P1_binary} ensures that $\tilde{\zeta_k}[n]\in\{0,1\},~\forall k, n$. The subproblem $(\mathrm{P1})$ is still non-convex due to non-convexity of the constraints \eqref{cvx_p1} and \eqref{P1_binary} and for general $N$, it is indeed NP-hard. Therefore, we cannot solve it efficiently.  To handle \eqref{cvx_p1}, we first present Lemma \ref{lemma1} below.

\begin{lemma}\label{lemma1}
Define the bivariant functions $Z_1(x,y; a, b)\triangleq x\ln(1+\frac{ay}{y+bx})$ and $Z_2(x, y; c) \treq x\ln(1+\frac{cy}{x})$ over the domain $x, y>0$ with the positive constants, i.e., $a, b, c > 0$. Both $Z_1$ and $Z_2$ are jointly concave w.r.t the variables $x$ and $y$. Additionally, the inequality below near the given point $(x_0,y_0)$ always holds with tightness:
\begin{align} \label{Z2_approx}
    Z_2(x, y; c) &\leq  x_0\big(1+c\frac{y_0}{x_0}\big) \nonumber\\
    &\hspace{-10mm}+ 
    \left(\ln\big(1+c\frac{y_0}{x_0}\big) - \frac{cy_0}{x_0 + cy_0}  \right) (x-x_0) \nonumber\\
     &\hspace{-10mm}+ \left( \frac{cx_0}{x_0 + cy_0}\right)(y-y_0) \treq f^{ub}_1(x,y; x_0, y_0, c),
\end{align}

\begin{proof}
Please see Appendix \ref{Appendix A}.
\end{proof}
\end{lemma}

Using Lemma \ref{lemma1}, it can be identified that both Terms I and II in \eqref{cvx_p1} are concave w.r.t the optimization variables $\tilde{\pmb{\zeta}}$ and $\tilde{\P}_k$, since the summation operator preserves the convexity. The non-convexity of the left-hand-side (LHS) expression is in the form of \textit{concave-minus-concave}. Then, using \eqref{Z2_approx} and applying the SCA technique, we approximate the non-convex constraint with the corresponding approximate convex one at each iteration. Similar restrictive approximation can also be applied to convert \eqref{P1_binary} into a convex constraint. Given the local point $(\tilde{\P}^{(l)}_k, \tilde{\pmb{\zeta}}^{(l)})$ in the $l$-th iteration, we can obtain a convex reformulation of problem $(\mathrm{P1})$ as follows.
\begin{subequations}\label{usrpowsch_subprob_cvx}
\begin{align}
(\mathrm{P1.1}):& \stackrel{}{\underset{\psi, \tilde{\pmb{\zeta}}, \tilde{\P}_k}{\mathrm{maximize}}~~\psi} \nonumber\\
\text{s.t.}~~~\begin{split} 
\sum^{N}_{n=1}&\tilde{\zeta}_k[n]
\ln\left(1+\frac{C_{n}\tilde{p}_k[n]}{\tilde{p}_k[n]+D_{k, n}\tilde{\zeta}_k[n]}\right)\\
&\hspace{-7mm}-f^{ub}_1(\tilde{\zeta}_k[n],\tilde{p}_k[n]; \tilde{\zeta}_{k, n}, p^{(l)}_{n, k}, B_{k,n}) \geq \frac{\psi}{\lambda_1},\quad\forall k  
 \end{split}  \label{cvx_p11}\\
&\hspace{-5mm} \sum^{N}_{n=1}\sum^{K}_{k=1}\left[(1- 2\tilde{\zeta}^{(l)}_{k}[n]) \tilde{\zeta}_k[n] +(\tilde{\zeta}^{(l)}_{k}[n])^2\right] \leq  0,\label{P11_binary_cvx}\\
&\eqref{19b} - \eqref{19e} \label{p11_last}
\end{align}  
\end{subequations}
Since the reformulated problem  $(\mathrm{P1.1})$ is convex w.r.t the optimization variables $\{\psi, \tilde{\pmb{\zeta}}, \tilde{\P}_k\}$, it can be solved efficiently via CVX using the interior-point method \cite{cvx} commencing from a feasible point. Nonetheless,  due to co-existence of constraints \eqref{19e} and \eqref{P11_binary_cvx}, attaining a feasible solution is, in general, difficult. Therefore, to embark on tackling this issue, we apply the penalty-SCA (PSCA) technique, wherein the constraint \eqref{P11_binary_cvx} is brought into the objective function as a penalty term \cite{zhou2021uav}. This initially violates the binary constraint but makes the problem feasible, enabling us to iteratively employ the SCA technique to solve the resulting convex optimization problem. To proceed, we can express  $(\P1.1)$ approximately as
\begin{subequations}\label{usrpowsch_subprob_cvx_p12}
\begin{align}
(\mathrm{P1.2}):& \stackrel{}{\underset{\eta, \psi, \tilde{\pmb{\zeta}}, \tilde{\P}_k}{\mathrm{maximize}}~~\psi - \mu \eta} \nonumber\\
\text{s.t.}&~~~\eqref{cvx_p11}, \eqref{p11_last}\\
&\hspace{-5mm} \sum^{N}_{n=1}\sum^{K}_{k=1}\left[(1- 2\tilde{\zeta}^{(l)}_{k}[n]) \tilde{\zeta}_k[n] +(\tilde{\zeta}^{(l)}_{k}[n])^2\right] \leq  \eta,\label{P11_binary_cvx_modified}
\end{align}
\end{subequations}
where $\eta$ is a non-negative slack variable, and $\mu$ is the given penalty parameter. It is worth stressing that feasible set of problem $(\mathrm{P1.2})$ is larger than that of $(\mathrm{P1})$; thus, the parameter $\mu$ should initially be chosen some small non-negative  value to make less focus on the binary constraint, then by gradually increasing $\mu$ and solving the convex problem $(\mathrm{P1.2})$ in an alternating manner, $\eta$ is forced to approach zero with some accuracy, and accordingly, the approximate solution of joint user scheduling and transmit power optimization can be achieved. To this end,  we can obtain the optimized value of $\P_k=\{{p}_k[n] = \frac{\tilde{p}_k[n]}{\tilde{\zeta}_k[n]}, \forall k, n\}$. Further, once the solution of the overall algorithm is obtained, we can reconstruct the corresponding binary solution of ${\pmb{\zeta}}$, according to  ${\pmb{\zeta}}=\{{\zeta}_k[n] = \lfloor{\tilde{{\zeta}}_k[n] + 0.5}\rfloor, \forall k, n\}$.

\begin{remark}
The formulated convex optimization model given in $(\mathrm{P1.1})$, though being convex, cannot be directly accepted by CVX, as it does not follow the disciplined convex programming (DCP) ruleset required. Given that the relative entropy function  $\mathbf{E}_{rel}(x,y)=x\log(\frac{x}{y}), x, y>0$ is convex and accepted by CVX, we can rewrite concave function $Z_1(x,y;a,b)$ (or the equivalent expression in the constraint \eqref{cvx_p1}), as
\begin{align}
    Z_1(x,y;a,b) &= 
    \frac{1}{b}\bigg[(y+bx)\ln\left(1+\frac{ay}{y+bx}\right)-\frac{1}{a}
    \nonumber\\
    &\hspace{31mm}\times ay\ln \left(1+\frac{ay}{y+bx}\right)\bigg]\nonumber\\
    &\stackrel{(a)}{=}  -\left(\frac{1+a}{ab}\right)\mathbf{E}_{rel}\Big(y+bx, (a+1)y+bx\Big)\nonumber\\
    &-\left(\frac{1}{ab}\right)\mathbf{E}_{rel}\Big((a+1)y+bx, y+bx\Big),
\end{align}
where the equality $(a)$ follows from the following relations between different form of logarithmic functions  and  the convex relative entropy function given by
\begin{align}
    x\ln\left(1+\frac{y}{x}\right) &= -\mathbf{E}_{rel}(x, x+y), \label{xlnyx}\\
    x\ln\left(1+\frac{x}{y}\right) &=  \mathbf{E}_{rel}(x+y, y) + \mathbf{E}_{rel}(y, x+y),\label{xlnxy}
\end{align}
 wherein \eqref{xlnyx} and \eqref{xlnxy} are jointly concave and convex w.r.t the joint variables $(x,y)$ over $x, y > 0$, respectively.
\end{remark}

In terms of computational cost, for a given $\mu$ we have $(2(NK+1))$ optimization variables and $(K(2N+1)+2)$ convex constraints in subproblem $(\mathrm{P1.2})$. Assume the convergence accuracy of SCA algorithm employed for solving this subproblem is  $\varepsilon_1$, the worst-case complexity of solving approximated subproblem $(\mathrm{P1.2})$ can be attained as $\O\left(U(2(NK+1))^2(K(2N+1)+2)^{1.5}\log_2(\frac{1}{\varepsilon_1})\right)$, where $U$ is the maximum number of iteration required until the outer loop of the two-layer alternating optimization subproblem $(\mathrm{P1.2})$ converges.

\subsection{UUR's relaying power optimization}
The corresponding subproblem for optimizing UUR's relaying power can be rewritten, introducing the slack variable $\psi$,  as
\begin{subequations}
\begin{align}\label{subprob_pu}
(\mathrm{P2}):& \stackrel{}{\underset{\psi, \P_u}{\mathrm{maximize}}~~\psi} \\
\text{s.t.}~~~
\begin{split} \sum^{N}_{n=1}\left[\lambda_{k, n}\ln\left(1+\frac{E_{k, n}p_u[n]}{p_u[n]+F_{k, n}}\right)\right] - G_{k}
 \geq \psi,\quad\forall k 
\end{split} \label{ncvx_p2}\\
&\frac{1}{N}\sum_{n=1}^{N}p_u[n] \leq p^{ave}_u,\\
&0 \leq p_u[n] \leq p^{max}_u,\forall n 
\end{align}
\end{subequations}
where $\lambda_{k, n} = \frac{\textcolor{black}{B}\tilde{\zeta}_k[n]}{2 \ln 2 \sum^{N}_{n=1} P_{f}[n]}$, $E_{k, n}= p_k[n]g_{ku}[n]$ and
\begin{alignat*}{2}\small
F_{k, n} &= \frac{p_k[n] g_{ku}[n]+p_b[n]g_{bu}[n]+1}{g_{bu}[n]} \\
G_{k} &= \frac{\sum^{N}_{n=1}\textcolor{black}{B}\tilde{\zeta}_k[n]\log_2\left(1+\frac{p_k[n] g_{ku}[n]}{p_b[n]g_{bu}[n] + 1}\right)}{2\sum^{N}_{n=1}P_{f}[n]} ~~\forall k, n
\end{alignat*}
Note that  subproblem $(\mathrm{P2})$ is a convex optimization problem due to having an affine objective function and all convex 
constraints, following from Lemma \ref{lemma2} introduced below. 
\begin{lemma}\label{lemma2}
Let define the function $f_2(x; a,b,c,d), x\geq 0$ with positive constant values $a, b, c, d > 0$ as 
\begin{align}
    f_2(x; a,b,c,d)&\triangleq\ln\left(1+\frac{ax+b}{cx+d}\right),\nonumber\\
    &=\ln\left(1+\frac{a}{c}-\frac{ad-bc}{c^2x+cd}\right),\label{reform_cvxdcp}
\end{align}
$f_2(x)$ is concave subject to the condition $ad\geq bc$, following from the fact that the function $\ln(1+qx), q \geq 0, x>0$ is concave w.r.t $x$, whose extended-value extension is non-decreasing and $h(x)=-\frac{1}{x}$ is also concave; therefore, $(f\circ g)(x)$ is concave. Note that the last equality of \eqref{reform_cvxdcp} represents the understandable reformulation of the function $ f_2(x; a,b,c,d)$ by the CVX optimization toolbox.
We also stress that for any given point $x_0$, there is a unique convex function $f^{lb}_2(x; x_0, a, b, c, d)$ defined as
\begin{align}
    f^{lb}_2(x; x_0, a, b, c, d) &\treq \ln(1+\frac{ax_0+b}{cx_0+d}) \nonumber\\
    &+ \frac{(ad-bc)(x-x_0)}{(cx_0+d)(b+d+(a+c)x_0)}.
\end{align}
such that $f^{lb}_2(x; x_0, a, b, c, d)$ serves as a global lower-bound of $f_2(x)$, i.e., $f_2(x) \geq f^{lb}_2(x; x_0, a, b, c, d)$ \cite{cvx_boyd}.
\end{lemma}
Consequently, one can solve subproblem $(\mathrm{P2})$ efficiently using CVX. Here, we have $(N+1)$ optimization variables and $(N+K+1)$ convex constraints. Assuming the convergence accuracy of interior-point algorithm employed for solving this convex problem with logarithmic cone is $\varepsilon_2$, the complexity cost of solving subproblem $(\mathrm{P2})$ can be obtained as $\O\left((N+1)^2(N+K+1)^{1.5}\log_2(\frac{1}{\varepsilon_2})\right)$.

\subsection{BS's jamming power optimization}
Keeping the other variables unchanged and taking the slack variable $\psi$, the BS's jamming power optimization subproblem is given as
\begin{subequations}
\begin{align}\label{subprob_pb}
(\mathrm{P3}):& \stackrel{}{\underset{\psi, \P_b}{\mathrm{maximize}}~~\psi} \\
\text{s.t.}~~~
\begin{split}\sum^{N}_{n=1}\lambda_{k, n}\bigg[
&\ln\left(1+\frac{H_{k, n}}{p_b[n]+I_{n}}\right) \\
&-\ln\left(1+\frac{J_{k, n}}{p_b[n]+K_{n}}\right)\bigg]
 \geq \psi,\quad\forall k 
\end{split} \label{ncvx_p3}\\
&\frac{1}{N}\sum_{n=1}^{N}p_b[n] \leq p^{ave}_b,\label{pbsum}\\
&0 \leq p_b[n] \leq p^{max}_b, \forall n \label{pbmax}
\end{align}
\end{subequations}
where $H_{k, n} = g_{ku}[n]p_k[n]p_u[n]$, $J_{k, n} = \frac{p_k[n]g_{ku}[n]}{g_{bu}[n]}$, $K_{n} = \frac{1}{g_{bu}[n]}$, and 
$I_{k, n} = \frac{p_u[n]g_{bu}[n]+p_k[n]g_{ku}[n]+1}{g_{bu}[n]}$.
Notice that subproblem $(\mathrm{P3})$ is non-convex due to non-convex constraint \eqref{ncvx_p3}, which is in the form of \textit{convex-minus-convex} according to  \cite[Lemma 1]{mamaghani2020improving}. Therefore, we apply SCA such that for a given local point $\P^{(l)}_b$ in $l$-th iteration. We approximate the first convex term with the global underestimator concave expression and obtain the convex reformulation as
\begin{subequations}
\begin{align}\label{subprob_pb_cvx}
(\mathrm{P3.1}):& \stackrel{}{\underset{\psi, \P_b}{\mathrm{maximize}}~~\psi} \\
\text{s.t.}~~~
\begin{split}\sum^{N}_{n=1}\lambda_{k, n}\bigg[
&f_{3}(p_b[n]; p^{(l)}_{n, b}, H_{k, n}, I_{k, n}) \\
&-\ln\left(1+\frac{J_{k, n}}{p_b[n]+K_{n}}\right)\bigg]
 \geq \psi,\quad\forall k 
\end{split} \label{cvx_p31}\\
&\eqref{pbsum}~\&~\eqref{pbmax}
\end{align}
\end{subequations}
wherein
\begin{align}
&f_{3}(p_b[n]; p^{(l)}_{n, b}, H_{k, n}, I_{k,n})=
\ln\left(1+\frac{H_{k, n}}{p^{(l)}_{n, b}+I_{k,n}} \right) \nonumber\\
&- \frac{H_{k, n}}{(p^{(l)}_{n, b}+I_{k,n})(p^{(l)}_{n, b}+H_{k, n}+I_{k,n})} (p_b[n]-p^{(l)}_{n, b}).   
\end{align}
Since subproblem $(\mathrm{P3.1})$ is convex, we can solve it efficiently using CVX. Here, we have $N+1$ optimization variables and $(N+K+1)$ convex constraints. Assuming the accuracy of SCA algorithm for solving this problem is  $\varepsilon_3$, the complexity of solving approximated subproblem $(\mathrm{P3.1})$ can, therefore, be represented as $\O\left((N+1)^2(N+K+1)^{1.5}\log_2(\frac{1}{\varepsilon_3})\right)$.

\subsection{Joint trajectory and velocity optimization}
Now, we optimize the trajectory $\q$ and velocity $\v$ of the UUR while keeping the transmit power allocation  and user scheduling sets ($\P$, $\pmb{\zeta}$) fixed. Therefore, the corresponding subproblem can be given as
\begin{subequations}
\begin{align}\label{subprob_qvu_ncvx}
(\mathrm{P4}):& \stackrel{}{\underset{\Q}{\mathrm{maximize}}~~\min_{k\in\K}\frac{\bar{R}^{k}_{sec}\left(\q, \v\right)}{\bar{P}_{f}(\v)}} \\
&\text{s.t.}~~~\mathrm{C9-C14}\label{c9c14}
\end{align}
\end{subequations}
wherein $\bar{P}_{f}(\v)=\frac{1}{N}\sum^{N}_{n=1}{P}_{f}[n]$. In order to solve subproblem $(\mathrm{P4})$, we should maximize every single fractional terms of $\left\{\frac{\bar{R}^{k}_{sec}\left(\q, \v\right)}{\bar{P}_{f}(\v)}, \forall k\right\}$ subject to the given constraint \eqref{c9c14}. In light of this, let $\lambda^\star$ be the maximum MSEE of subproblem
$(\mathrm{P4})$ with solution set $\left(\q^\star, \v^\star\right)$ given by
\begin{align}
    \lambda^\star &= \max_{\q, \v \in \mathcal{F}}\min_{k\in\K}\frac{\bar{R}^{k}_{sec}\left(\q, \v\right)}{\bar{P}_{f}(\v)} \nonumber\\
    &= \min_{k\in\K}\frac{\bar{R}^{k}_{sec}\left(\q^\star, \v^\star\right)}{P_{f}(\v^\star)},
\end{align}
wherein $\mathcal{F}$ represents the feasible set spanned by the constraint \eqref{c9c14}. Applying nonlinear fractional Dinkelbach programming theory \cite{dinkelbach1967nonlinear}, the objective function of problem $(\mathrm{P4})$ can be equivalently transformed into a subtractive version such that the optimal value of $\lambda^\star$ can be achieved iff
\begin{align}
  \max_{\q, \v \in \mathcal{F}}&\min_{k\in\K}\bar{R}^{k}_{sec}\left(\q, \v\right) - \lambda^\star{\bar{P}_{f}(\v)}\nonumber\\
  &= \min_{k\in\K}\bar{R}^{k}_{sec}\left(\q^\star, \v^\star\right) - \lambda^\star{\bar{P}_{f}(\v^\star)}=0, 
\end{align}
Thus, we can optimize the equivalent problem to obtain the optimal solution of $\Q$, via solving the reformulated problem as 
\begin{subequations}
\begin{align}\label{prob51}
(\mathrm{P4.1}):& \stackrel{}{\underset{\q, \v}{\mathrm{maximize}}~~\min_{k\in\K}\bar{R}^{k}_{sec}(\q, \v) - \lambda^{(m)}\bar{P}_{f}(\v)} \\
&\text{s.t.}~~~ \eqref{c9c14}
\end{align}
\end{subequations}
wherein $\lambda^{(m)} = \min_{k\in\K}\frac{\bar{R}^{k}_{sec}\left(\q^{(m)}, \v^{(m)}\right)}{{\bar{P}_{f}(\v^{(m)})}}$ showing the value of $\lambda$ in the $m$-th iteration of the Dinkelbach algorithm. Reformulated problem $(\mathrm{P4.1})$ is still non-convex due to non-convex objective function and constraint $\mathrm{(C9)}$ which can be dealt with as follows. 

By introducing the slack variables $\psi$ and $\pmb{\mu}=\{\mu[n]\}^N_{n=1}$ such that 
\begin{align}
    \mu[n] = \left(\sqrt{1+\frac{\|\v[n]\|^4}{4\nu^4_0}} - \frac{\|\v[n]\|^2}{2\nu^2_0}\right)^{\frac{1}{2}}, \quad\forall n
\end{align}
we can relax the problem $(\mathrm{P4.1})$ to the one with the approximately equivalent but enjoying concave objective function as
\vspace{-5mm}
\begin{subequations}
\begin{align}\label{prob52}
(\mathrm{P4.2}):& \stackrel{}{\underset{\psi, \pmb{\mu}, \q, \v}{\mathrm{maximize}}~~\psi - \lambda^{(m)}\omega} \\
&\text{s.t.}~~~\mathrm{C10-C14} \label{c10c14}\\
&\omega \leq \bar{P}_{lim}, \label{flight_cvx} \\
&\mu[n]\geq 0, \quad\forall n \label{mu_pos}\\
& \mu^2[n]+ \frac{\|\v[n]\|^2}{\nu^2_0} \geq \frac{1}{\mu^2[n]}, \quad\forall n  \label{prop_ncvx}\\
&   \bar{R}^{k}_{sec} \geq \psi , \quad\forall k \label{rsec_ncvx}
\end{align}
\end{subequations}
wherein $\omega \treq \bar{\P}^{ub}_{f}(\v)=\frac{1}{N}\sum^{N}_{n=1}{P}^{ub}_{f}[n]$, with $\{{P}^{ub}_{f}[n], \forall n\}$ serving as a global convex upper-bound of \eqref{flightpow}, defined as
\begin{align}
    {P}^{ub}_{f}[n] \hspace{-1mm}=\hspace{-1mm} {P_0\left(\hspace{-0.5mm}1\hspace{-1mm}+\hspace{-1mm}\frac{2\|\v[n]\|^2}{\Omega^2_uR_u^2}\hspace{-0.5mm}\right)} \hspace{-1mm}+\hspace{-1mm}{\frac{1}{2}d_0\rho sA \|\v[n]\|^3} \hspace{-1mm}+\hspace{-1mm} P_i\mu[n], 
\end{align}
Note that constraint \eqref{prop_ncvx} must be met with equality at the optimal point, because $\mu[n]$ can be otherwise decreased, resulting in an increase of the value of the objective function, which of course, violates the optimality. Plus, we also point out that the objective function, the constraints $\mathrm{C10-C14}$, and \eqref{flight_cvx} are now convex. However, the problem $(\mathrm{P4.2})$ is still unsolvable due to the generated extra non-convex constraints \eqref{prop_ncvx} and \eqref{rsec_ncvx}. Note that the LHS expression of \eqref{prop_ncvx}; i.e., summation of norm-square components, is jointly convex w.r.t the variables $\mu[n]$ and $\v[n]$. Owing to the fact that the right-hand-side (RHS) of \eqref{prop_ncvx} is convex, since the  second derivative of the inverse-square function $\frac{1}{\mu^2[n]}$ is non-negative; therefore, by replacing the LHS with the corresponding global concave lowerbound using first-order Taylor expansion at the local given point  $(\mu^{(m)}_{n}, \v^{(m)}_{n})$ with superscript $m$ indicating the iteration index of fractional Dinkelbach programming, we can reach the approximate convex constraint, associated with \eqref{prop_ncvx}, as
\begin{align}\label{prop_cvx}
  -&(\mu^{(m)}_n)^2 + 2\mu^{(m)}_n\mu[n]+ \frac{1}{v^2_0}  \nonumber\\
  &\times\left( -\|\v^{(m)}_n\|^2 + 2\v^{(m)}_n\v^\dagger[n]\right)\geq \frac{1}{\mu^2[n]}, \quad\forall n 
\end{align}
Now, we deal with the last non-convex constraint \eqref{rsec_ncvx} by introducing the slack variables $\s=\{s_k[n], \forall k, n\}$, $\r=\{r_k[n], \forall k, n\}$, and $\w=\{w[n], \forall n\}$, rewriting problem $(\mathrm{P4.2})$ as
\begin{subequations}
\begin{align}\label{prob43}
(\mathrm{P4.3}):& \stackrel{}{\underset{\psi, \pmb{\mu}, \q, \v, \s, \r, \w}{\mathrm{maximize}}~~\psi - \lambda^{(m)}\omega} \\
&\hspace{-7mm}\text{s.t.}~~~\eqref{c10c14},~\eqref{flight_cvx},~\eqref{mu_pos},~\eqref{prop_cvx} \label{36b}\\
\begin{split}
&\hspace{-7mm}\frac{\textcolor{black}{B}}{2N\ln 2} \sum^N_{n=1}\tilde{\zeta}_{k, n} \bigg[\ln\left(1+\frac{1}{k_0r_k[n]+k_1w[n] + \epsilon}\right)\\
&\hspace{-7mm}-\ln \left(1+\frac{k_2s^{-1}_k[n]}{k_3w^{-1}[n]+1}\right)\bigg] \geq \psi,\quad\forall k
\end{split} \label{36c}\\
 \begin{split}
&\hspace{-7mm}\frac{N_0}{\beta_0}\left(\|\q[n]-q_k\|^2 + H^2\right)\times\\
&\hspace{-7mm}\exp(a_f \sqrt{\left(\|\q[n]-q_k\|^2 + H^2\right)}) \geq s_k[n], \quad\forall k, n\label{sk_lb}\\
\end{split}\\
 \begin{split}
&\hspace{-7mm}\frac{N_0}{\beta_0}\left(\|\q[n]-q_k\|^2 + H^2\right)\times\\
&\hspace{-7mm}\exp(a_f \sqrt{\left(\|\q[n]-q_k\|^2 + H^2\right)}) \leq r_k[n], \quad\forall k, n\label{rk_ub}\\
\end{split}\\
\begin{split}
&\hspace{-7mm}\frac{N_0}{\beta_0} \left(\|\q[n]-q_b\|^2+H^2\right)\times\\
&\hspace{-7mm}\exp(a_f \sqrt{\left(\|\q[n]-q_b\|^2+H^2\right)}) \leq w[n], \quad\forall n \label{cvx_36e}
\end{split}
\end{align}
\end{subequations}
where in \eqref{36c}, we have defined $k_0 = \frac{p_u[n]+p_b[n]}{p_k[n]p_u[n]}$, $k_1=\frac{1}{p_u[n]}$,$k_2 = p_k[n]$,$k_3 = p_b[n]$, $\epsilon = 1/(p_k[n]p_u[n]g_{ku}[n]g_{bu}[n])$.
Note that all the inequality constraints \eqref{sk_lb}, \eqref{rk_ub}, and \eqref{cvx_36e} must also be met with equality at the optimal point, otherwise the optimality is violated. Following the high-SNR approximation, we set $\epsilon \approx 0$ in the subsequent sections for the ease of expositions. We remark the fruitful lemma below.
\begin{lemma}\label{lemma3}
Let define the bivariate functions $f_{41}(x, y; a, b)$ and $f_{42}(x, y; c, d)$, and univariate functions $f_{43}(x; e)$ and $f_{44}(x; e)$ with positive constants $a, b, c, d, p, r >0$ as
\begin{align*}
f_{41}(x, y; a, b) \hspace{-1mm}&=\hspace{-1mm} \ln\left(\hspace{-1mm}1\hspace{-1mm}+\hspace{-1mm}\frac{1}{ax\hspace{-1mm}+\hspace{-1mm}by}\hspace{-1mm}\right),~~~~~~~f_{43}(x; p) \hspace{-1mm}=\hspace{-1mm} x^2\exp(px),\nonumber\\
f_{42}(x, y; c, d) \hspace{-1mm}&=\hspace{-1mm} \ln\left(1\hspace{-1mm}+\hspace{-1mm}cx^{-1}\hspace{-1mm}+\hspace{-1mm}dy^{-1}\right),~~f_{44}(x; r) \hspace{-1mm}=\hspace{-1mm} \ln\left(1\hspace{-1mm}+\hspace{-1mm}\frac{r}{x}\right).
\end{align*}
We have the following tight inequalities
\begin{align}
    f_{41}(x, y)  &\geq f_{41}(x_0, y_0) \hspace{-1mm}-\hspace{-1mm}\frac{a(x\hspace{-1mm}-\hspace{-1mm}x_0)}{{\left(a\,x_0\hspace{-1mm}+\hspace{-1mm}b\,y_0\right)}\,{\left(a\,x_0\hspace{-1mm}+\hspace{-1mm}b\,y_0\hspace{-1mm}+\hspace{-1mm}1\right)}} \nonumber\\
    &\hspace{-12mm}-\frac{b(y\hspace{-1mm}-\hspace{-1mm}y_0)}{{\left(a\,x_0\hspace{-1mm}+\hspace{-1mm}b\,y_0\right)}\,{\left(a\,x_0\hspace{-1mm}+\hspace{-1mm}b\,y_0\hspace{-1mm}+\hspace{-1mm}1\right)}} \treq f^{lb}_1(x, y; x_0, y_0, a, b),
\end{align}
\begin{align}
    f_{42}(x, y) &\geq f_{42}(x_0, y_0) -\frac{c\,y_0 (x\hspace{-1mm}-\hspace{-1mm}x_0)}{x_0\,{\left(c\,y_0\hspace{-1mm}+\hspace{-1mm}d\,x_0\hspace{-1mm}+\hspace{-1mm}x_0\,y_0\right)}}
    \nonumber\\
    &\hspace{-12mm}-\frac{d\,x_0 (y\hspace{-1mm}-\hspace{-1mm}y_0)}{y_0\,{\left(c\,y_0\hspace{-1mm}+\hspace{-1mm}d\,x_0\hspace{-1mm}+\hspace{-1mm}x_0\,y_0\right)}}\treq f^{lb}_{42}(x, y; x_0, y_0, c, d),
    \end{align}
  \begin{align}
 f_{43}(x)  &\geq f_{43}(x_0)\hspace{-1mm}+\hspace{-1mm} x_0\,{\mathrm{e}}^{p\,x_0} \,{\hspace{-1mm}\left(p\,x_0\hspace{-1mm}+\hspace{-1mm}2\right)} (x\hspace{-1mm}-\hspace{-1mm}x_0)\hspace{-1mm}\nonumber\\
 &\treq \hspace{-1mm} f^{lb}_{43}(x; x_0, p),\\
    f_{44}(x) &\geq f_{44}(x_0) \hspace{-1mm}-\hspace{-1mm} \frac{r(x\hspace{-1mm}-\hspace{-1mm}x_0)}{x_0(x_0\hspace{-1mm}+\hspace{-1mm}r)} \treq f^{lb}_{44}(x; x_0, r),
\end{align}
\end{lemma}
\begin{proof}
Please see Appendix \ref{Appendix B}.
\end{proof}
By introducing the slack variables $\u=\{u_k[n], \forall k, n\}$, and using Lemma \ref{lemma3}, we can approximate the non-convex problem $(\mathrm{P4.3})$ with a more tractable reformulation  given as
\begin{subequations}
\begin{align}\label{prob44}
(\mathrm{P4.4}):& \stackrel{}{\underset{\psi, \pmb{\mu}, \q, \v, \s, \r, \w, \u}{\mathrm{maximize}}~~\psi - \lambda^{(m)}\omega} \\
&\hspace{-12mm}\text{s.t.}~~~\eqref{36b},~\eqref{rk_ub},~\eqref{cvx_36e} \label{45b}\\
\begin{split}
     &\hspace{-12mm}\frac{\textcolor{black}{B}}{2N\ln 2} \sum^N_{n=1}\tilde{\zeta}_{k, n} \bigg[
     f^{lb}_{41}(r_k[n], w[n]; {r}^{(m)}_{k, n}, {w}^{(m)}_n, k_0, k_1)\\
     &\hspace{-12mm}-\ln\left(1+k_2s^{-1}_k[n]+k_3w^{-1}[n]\right)\\
     &\hspace{-12mm}+ f^{lb}_{44}(w[n]; {w}^{(m)}_n, k_3)\bigg] \geq \psi, \quad\forall k \label{45c}
\end{split}\\
&\hspace{-12mm}\frac{N_0}{\beta_0} f^{lb}_{43}(u_k[n]; u^{(m)}_{k, n}, a_f) \geq s_k[n], \quad\forall k, n\label{45d}\\
&\hspace{-12mm} \sqrt{\|\q[n]-\q_k\|^2 + H^2}\geq u_k[n], \quad\forall k, n
\end{align}
\end{subequations}
wherein $\{{r}^{(m)}_{k, n}, {w}^{(m)}_n,  u^{(m)}_{k, n}, \forall k, n\}$ are the value set of slack variables $(\r, \w, \u)$ in the $m$-th iteration of Dinkelbach algorithm. Finally, since the last constraint is non-convex, we apply \cite[Lemma 3]{tatarmamaghani2021joint} to approximate it with the corresponding convex constraint using the SCA approach, and obtain an approximate convex reformulation of $(\mathrm{P4.4})$ as
\begin{subequations}
\begin{align}\label{prob45}
(\mathrm{P4.5}):& \stackrel{}{\underset{\psi, \pmb{\mu}, \q, \v, \s, \r, \w, \u}{\mathrm{maximize}}~~\psi - \lambda^{(m)}\omega} \\
&\text{s.t.}~~~\eqref{45b},~\eqref{45c},~\eqref{45d}\\
\begin{split}
&-\|\q^{(m)}_n\|^2 + 2\left(\q^{(m)}_n - q_k\right)^\dagger\q[n] \\
&+ \|q_k\|^2 +H^2
\geq u^2_k[n], \quad\forall k, n
\end{split}
\end{align}
\end{subequations}
wherein $\{{\q}^{(m)}_{n}, \forall n\}$ is the local given point set of optimization variables $\q$ in the $m$-th iteration. Since subproblem $(\mathrm{P4.5})$ is convex; therefore, it can be efficiently solved via CVX.
It is worth noting that to solve subproblem $(\mathrm{P4.5})$, we have $(3N(K+2)+1)$ optimization variables and $(3NK+7N+K+1)$ convex constraints. Assuming the accuracy of SCA algorithm for solving this problem is  $\varepsilon_4$, the complexity of solving approximated subproblem $(\mathrm{P4.5})$ for given $\lambda^{(m)}$ can, therefore, be obtained as $\O\left((3N(K+2)+1)^2(3NK+7N+K+1)^{1.5}\log_2(\frac{1}{\varepsilon_4})\right)$.

\begin{remark}
Note that constraints given by \eqref{rk_ub} and \eqref{cvx_36e}, being in the form of 
$ a\|\mathbf{x}-\mathbf{x}_0\|^2\exp(b\|\mathbf{x}-\mathbf{x}_0\|) \geq y$, plus, the expression $\mathcal{E}=\ln(1+cx^{-1}+dy^{-1})$ used in \eqref{45c} are proved to be convex; however, they indeed violate the DCP rule-set of the CVX, and so cannot be applied in the optimization model.  The former can be handled by rewriting it as
\begin{align}
    t_1 \geq \|\mathbf{x}-\mathbf{x}_0\|^2,~ t_2 + a^{-1} \mathbf{E}_{rel}(at_1, y) \leq 0,~
    t_2 \geq b t^{\frac{3}{2}}_1,
\end{align}
And the latter can be dealt with properly by replacing $\mathcal{E}$-form function appeared in \eqref{45c} with $t_5$ and adding the constraints 
\begin{align}
    \frac{x}{c}\geq \exp(-t_3),~
    \frac{y}{d}\geq \exp(-t_4),~
    t_5 \geq \mathbf{LSE}(0, t_3, t_4),
\end{align}
wherein $t_1-t_5$ are some non-zero slack variables, and the log-sum-exp function, which is a CVX-approved convex function, defined as $\mathbf{LSE}(x_1, x_2, \cdots, x_n)=\ln(\sum^{N}_{i=1}\exp(x_i))$.
\end{remark}

\begin{figure}[!t]
\resizebox{\columnwidth}{!}{%
 \removelatexerror
\begin{algorithm}[H]
\SetAlgoLined
\KwResult{$\q^\star$, $\v^\star$}
\textbf{Initialize} feasible point $(\q^{(0)}, \v^{(0)})$ and slack variables, set iteration index $m=0$,  then $\psi^{(m)}=\bar{R}^{k}_{sec}\left(\q^{(m)},  \v^{(m)}\right)$, $\psi^{(m)} = \bar{P}_{f}( \v^{(m)})$, define $\lambda^{(m)}\treq\frac{\psi^{(m)}}{\omega^{(m)}}$, and set \textit{Convergence} = false;\\
 \While{not \textit{Convergence}}{
  Given $\left(\lambda^{(m)}\hspace{-1mm}, \q^{(m)}\hspace{-1mm}, \v^{(m)}\right)$, solve $(\mathrm{P4.5})$ using  \eqref{prob45}, then obtain $\left(\psi^{(m+1)}, \omega^{(m+1)}, \q^{(m+1)}, \v^{(m+1)}\right)$;\\
  Calculate $\lambda^{(m+1)}$, then $F =\psi^{(m)}- \lambda^{(m+1)}\omega^{(m)}$;\\
  \If{$|F| \leq \epsilon_2$}{
   $\q^\star= \q^{(m+1)}$, $\v^\star = \v^{(m+1)}$\;
   \textit{Convergence} = true\;}
   $m \gets m + 1$;\\
 } 
 \caption{Proposed Dinkelbach-based algorithm to approximately solve subproblem $\mathrm{(P4)}$} \label{Dinkelbach_algo}
\end{algorithm}}
\end{figure}

\begin{figure}[!t]
\resizebox{\columnwidth}{!}{%
 \removelatexerror
 \resizebox{\columnwidth}{!}{%
  \begin{algorithm}[H]
  \caption{Overall sequential based proposed iterative algorithm for MSEE maximization (MSEE-Seq)}\label{myalgo1}
  1:~\textbf{Initialize}~
  a feasible point ($\q^i, \v^i, \P^i_u, \P^i_b, \P^i_k, \pmb{\zeta}^i$),  and let iteration index $l=0$;\\
  2:~\textbf{Repeat:}\\
  3:~Solve $(\mathrm{P1.1})$ using \eqref{usrpowsch_subprob_cvx}, updating $\mathbf{P}^{(l+1)}_{k}\hspace{-1mm}$ and $\pmb{\zeta}^{(l+1)}$;~$l \gets l+1$;\\
  4:~Given $\left(\mathbf{P}^{(l+1)}_{k}\hspace{-1mm},\pmb{\zeta}^{(l+1)}\hspace{-1mm}\right)$, solve $(\mathrm{P2})$ using \eqref{subprob_pu}, updating $\P^{(l+1)}_u\hspace{-1mm}$;~$l \gets l+1$;\\
  5:~Given $\left(\mathbf{P}^{(l+1)}_{k},\pmb{\zeta}^{(l+1)},\P^{(l+1)}_u\right)$, solve $(\mathrm{P3.1})$ using \eqref{subprob_pb_cvx}, updating $\P^{(l+1)}_b$;~$l \gets l+1$;\\
  6:~Given $\left(\mathbf{P}^{(l+1)}_{k}, \pmb{\zeta}^{(l+1)}, \P^{(l+1)}_u, \P^{(l+1)}_b\right)$, run Algorithm \ref{Dinkelbach_algo} with $\q^{(l)}$ and $\v^{(l)}$, updating $\q^{(l+1)} \gets \q^\star$ and $\v^{(l+1)} \gets \v^\star$;~$l \gets l+1$;\\
  7:~\textbf{Until} fractional increase of objective function in \eqref{opt_prob} gets below the threshold $\epsilon_1$;\\
  8:~\textbf{Return:} $\left(\Q^{opt}\hspace{-1mm},\P^{opt}\hspace{-1mm},\pmb{\zeta}^{opt}\right)$ $\hspace{-1mm}\gets\hspace{-1mm}$ $\left(\Q^{(l)}\hspace{-1mm},\P^{(l)}\hspace{-1mm},\pmb{\zeta}^{(l)}\right)$;
  \end{algorithm}}}
\end{figure}

\begin{figure}[!t]
\resizebox{\columnwidth}{!}{%
 \removelatexerror
  \begin{algorithm}[H]
  \SetKwBlock{DoParallel}{Do in parallel}{end}
  \caption{Overall greedy based proposed iterative algorithm for MSEE maximization (MSEE-MI)}\label{myalgo2}
  1:~\textbf{Initialize}~
  a feasible point ($\q^i, \v^i, \P^i_u, \P^i_b, \P^i_k, \pmb{\zeta}^i$),  and let iteration index $l=0$;\\
  2:~\textbf{Repeat:}\\
  3:~\DoParallel{
    3.1:~Solve $(\mathrm{P1.1})$ using \eqref{usrpowsch_subprob_cvx} with $(\P^{(l)}_k, \pmb{\zeta}^{(l)})$\;
    3.2:~Solve $(\mathrm{P2})$ using \eqref{subprob_pu} with $\P^{(l)}_u$\;
    3.3:~Solve $(\mathrm{P3.1})$ using \eqref{subprob_pb_cvx} with $\P^{(l)}_b$\;
    3.4:~Run Algorithm \ref{Dinkelbach_algo} with $(\q^{(l)}, \v^{(l)})$\;
    }
  4:~Update one of the blocks ($\P^{(l+1)}_k, \pmb{\zeta}^{(l+1)}$), $\P^{(l+1)}_u$, $\P^{(l+1)}_b$, or ($\q^{(l+1)},  \v^{(l+1)}$)  whose maximum improvement of objective function given in \eqref{opt_prob} gets the highest, and keep the remained blocks unchanged\;
  5:~$l \gets l+1$;\\
  6:~\textbf{Until} fractional increase of objective function in \eqref{opt_prob} gets below the threshold $\epsilon_1$;\\
  7:~\textbf{Return:} $\left(\Q^{opt}\hspace{-1mm},\P^{opt}\hspace{-1mm},\pmb{\zeta}^{opt}\right)$ $\hspace{-1mm}\gets\hspace{-1mm}$ $\left(\Q^{(l)}\hspace{-1mm},\P^{(l)}\hspace{-1mm},\pmb{\zeta}^{(l)}\right)$;
  \end{algorithm}}
\end{figure}

\subsection{Overall algorithms and complexity discussion}
Having obtained an efficient optimization model for each subproblem in the previous section, we are now ready to propose iterative algorithms based on sequential block optimization and maximum improvement (MI) or the so-called greedy optimization introduced in \cite{nutini2017let}, summarized in Algorithm \ref{myalgo1} and Algorithm \ref{myalgo2}, respectively. The former is simpler to implement and requires less computations at each iteration. The latter converges faster thanks to a large step-size at each iteration and implementation via parallel computation capability; otherwise, it maybe too expensive. 

\textcolor{black}{
It can be mathematically proved that both algorithms are guaranteed to converge to at least a suboptimal solution. Particularly, for convergence analysis of Algorithm \ref{myalgo1}, let define the objective values of the original problem ($\P$), the subproblems ($\P2$) and ($\P3.1$)  at iteration $l$ as 
$\mathbf{MSEE}\left(\pmb{\zeta}^l, \mathbf{P}^l_\mathbf{k}, \mathbf{P}^l_\mathbf{u}, \mathbf{P}^l_\mathbf{b}, \mathbf{Q}^l\right)$, $\mathbf{\Theta}\left(\pmb{\zeta}^l, \mathbf{P}^l_\mathbf{k}, \mathbf{P}^l_\mathbf{u}, \mathbf{P}^l_\mathbf{b}, \mathbf{Q}^l\right)$, and $\mathbf{\Xi}\left(\pmb{\zeta}^l, \mathbf{P}^l_\mathbf{k}, \mathbf{P}^l_\mathbf{u}, \mathbf{P}^l_\mathbf{b}, \mathbf{Q}^l\right)$, respectively. Now, we can proceed as 
{\small
\begin{align}\label{convergence}
&\mathbf{MSEE}\left(\pmb{\zeta}^l, \mathbf{P}^l_\mathbf{k}, \mathbf{P}^l_\mathbf{u}, \mathbf{P}^l_\mathbf{b}, \mathbf{Q}^l\right) \nonumber\\
&\stackrel{(a)}{\leq} \mathbf{MSEE}\left(\pmb{\zeta}^{(l+1)}, \mathbf{P}^{(l+1)}_\mathbf{k}, \mathbf{P}^l_\mathbf{u}, \mathbf{P}^l_\mathbf{b}, \mathbf{Q}^l\right)\nonumber\\
&\stackrel{(b)}{=} \mathbf{\Theta}\left(\pmb{\zeta}^{(l+1)}, \mathbf{P}^{(l+1)}_\mathbf{k}, \mathbf{P}^{(l+1)}_\mathbf{u}, \mathbf{P}^l_\mathbf{b}, \mathbf{Q}^l\right)\nonumber\\
&\stackrel{(c)}{\leq} \mathbf{MSEE}\left(\pmb{\zeta}^{(l+1)}, \mathbf{P}^{(l+1)}_\mathbf{k}, \mathbf{P}^{(l+1)}_\mathbf{u}, \mathbf{P}^l_\mathbf{b}, \mathbf{Q}^l\right)\nonumber\\
&\stackrel{(d)}{=} \mathbf{\Xi}\left(\pmb{\zeta}^{(l+1)}, \mathbf{P}^{(l+1)}_\mathbf{k}, \mathbf{P}^{(l+1)}_\mathbf{u}, \mathbf{P}^l_\mathbf{b}, \mathbf{Q}^l\right)\nonumber\\
&\stackrel{(e)}{\leq} \mathbf{MSEE}\left(\pmb{\zeta}^{(l+1)}, \mathbf{P}^{(l+1)}_\mathbf{k}, \mathbf{P}^{(l+1)}_\mathbf{u}, \mathbf{P}^{(l+1)}_\mathbf{b}, \mathbf{Q}^l\right)\nonumber\\
&\stackrel{(f)}{\leq} \mathbf{MSEE}\left(\pmb{\zeta}^{(l+1)}, \mathbf{P}^{(l+1)}_\mathbf{k}, \mathbf{P}^{(l+1)}_\mathbf{u}, \mathbf{P}^{(l+1)}_\mathbf{b}, \mathbf{Q}^{(l+1)}\right),
\end{align}}
where $(a)$ follows from the definition of the optimal solution to the optimization subproblems ($\P1.2$). Besides, the equalities $(b)$ and $(d)$  hold on the grounds that the first order Taylor approximation is adopted via the SCA technique and that the objective function of ($\P2$) and ($\P3.1$) share the same value with the original function at the given point. Furthermore, $(c)$ and $(e)$ hold since the objective value of ($\P2$) and ($\P3.1$)  are tight lower-bound to that of the original problem. Finally, $(f)$ follows from Algorithm \ref{Dinkelbach_algo} whose convergence has been well understood \cite{dinkelbach1967nonlinear}. From the last inequality in \eqref{convergence}, it can be concluded that Algorithm \ref{myalgo1} is guaranteed to converge, since the feasible solution set of $\mathrm{(P)}$ is compact and its objective value is non-decreasing over iteration index $l$ and that the optimal value of MSEE is upper bounded by a finite value from the communications engineering perspective. Convergence proof for Algorithm \ref{myalgo2}  follows the similar approach and hence omitted for brevity.}
\textcolor{black}{
In terms of computational complexity, given $L$ and $M$ be the maximum convergence iteration of the outer overall BCD-SCA algorithm and the inner fractional sub-algorithm, and based on the convergence analysis of each subproblem given in the previous subsections, Algorithms \ref{myalgo1} and \ref{myalgo2} have the overall worst-case complexity of approximately  $\O\left(L(NK)^{3.5}\left(U\log_2(\frac{1}{\tilde{\varepsilon}_1}) \hspace{-1mm}+\hspace{-1mm} M \log_2(\frac{1}{\tilde{\varepsilon}_4})\right) \hspace{-1mm}+\hspace{-1mm} LN^{3.5}\log_2(\frac{1}{\varepsilon_2\varepsilon_3}) \right)$ and $\O\left(L(NK)^{3.5}\max\left(U\log_2(\frac{1}{\tilde{\varepsilon}_1}) , M \log_2(\frac{1}{\tilde{\varepsilon}_4}) \right)\right)$, respectively. Both are in polynomial time order and applicable to the energy-hungry UUR scenarios.}

\section{Numerical results and discussion}\label{sec:numerical}
In this section, we provide some numerical simulations to evaluate the secrecy performance of the considered THz-UUR scheme, and demonstrate the effectiveness of our proposed designs in comparison with some benchmarks. Unless otherwise stated, all simulation parameters, adopted from the literature, are given in Table \ref{table:1}.

Since the initial feasible point is important to use the proposed BCD-SCA-Dinkelbach based algorithms and significantly impacts their convergence performance, we explain how we can obtain an initial feasible UAV's trajectory and velocity, network transmission powers, and user schedulings. The initial UAV's trajectory is assumed to be a circular path centered at the BS's location with radius $R_u=\|\q_b-\q_I\|$, provided that UAV's instantaneous velocity constraint $\mathrm{C12}$ is satisfied, and $T \geq T^{min}_{cir} \treq \frac{2 \pi R_u}{V_{max}}$, where $T^{min}_{cir}$ is the minimum required time for circular trajectory. However, if $T^{min}_{cir} > T \geq T^{min}_{cyc} \treq \frac{2R_u}{V_{max}}$ (i.e., at least cyclic trajectory was possible with minimum required time $T^{min}_{cyc}$), then one could use any cyclic shape as long as $\mathrm{C10-C14}$ are satisfied. Here, we consider a Piriform trajectory with discretized equations given by $\q^i=[\mathbf{x}^i; \mathbf{y}^i]$ with $\mathbf{y}^i = A_y(1-\sin(\mathbf{t}))\cos(\mathbf{t})$ and $\mathbf{x_u} = R_u({\sin(\mathbf{t})+1})/{2}$ in which $\mathbf{t}_{1\times N}$ indicates the linearly spaced vector in $[\frac{\pi}{2}, \frac{5\pi}{2}]$. Further, the constant $A_y$ can be obtained efficiently via a simple 1D search in the range of $[R_u, 0]$ or simply set to zero. The UAV's initial velocity vector $\v^i$ is then followed  by $\v^i[n]=\frac{\q^i[n+1] - \q^i[n]}{\delta_t}, \forall n \setminus N$ and $\v^i[N] = \v^i[N-1]$. 

Having obtained an initial feasible UAV's path planning ($\Q^i=\{\q^i, \v^i\}$), we can set the initial UUR's relaying power and the BS's jamming transmission power as $\P^i_u = \{p^i_u[n]=p^{ave}_u, \forall n\}$ and $\P^i_b = \{p^i_b[n]=p^{ave}_b, \forall n\}$, respectively.
The UEs' initial transmit powers are set as 
\[\P^i_k =\left\{ p^i_k[n]=
\begin{cases}
p^{ave}_{k}, & \zeta_k[n]=1\\
0, & \zeta_k[n]=0,
\end{cases}~~\forall k, n\right\}\]
and the UEs are scheduled equally (e.g., $\lfloor\frac{N}{K}\rfloor$ times each), i.e.,  $\pmb{\zeta}^i$ is obtained such that the constraint $\mathrm{C1}$ holds. 

After identifying the initial feasible point for the iterative optimization algorithms, we consider different benchmark schemes, all of which are detailed below and labelled in the following figures,  to demonstrate the superiority of our proposed MSEE-based optimization algorithms.

\begin{itemize}
    \item \textbf{MSEE-Seq}: \textbf{M}inimum \textbf{S}ecrecy \textbf{E}nergy \textbf{E}fficiency optimization scheme using the \textbf{Seq}uential BCD-based subproblem maximization as given in Algorithm \ref{myalgo1}.
    
     \item \textbf{MSEE-MI}: \textbf{M}inimum \textbf{SEE} optimization scheme based on the \textbf{M}aximum \textbf{I}mprovement subproblem maximization as given in Algorithm \ref{myalgo2}.
     
     \item \textbf{MSEE-FTrj}: \textbf{M}inimum \textbf{SEE} design with \textbf{F}ixed \textbf{Tr}a\textbf{j}ectory and velocity, i.e., $\Q=\{\q^i[n], \v^i[n], \forall n\}$ based \textbf{M}inimum \textbf{SEE} optimization scheme using the MI-BCD approach via jointly optimizing the transmit power allocations and user scheduling, i.e., $\P_k$, $\P_u$, $\P_b$, and $\pmb{\zeta}$.
     
     \item \textbf{MSEE-FPow}: Jointly designing the trajectory and velocity of the UUR for the \textbf{M}inimum \textbf{SEE} design via solving the corresponding subproblem, i.e., optimizing $\Q$, while keeping the power allocations and user scheduling parameters fixed, i.e., setting them equal to the initial feasible values.
     
     \item \textbf{MASR-Seq}: Optimizing \textbf{M}inimum \textbf{ASR} given in \eqref{ave_rsec} while ignoring the UUR's flight power limit using the \textbf{Seq}uential BCD approach to iteratively improve $\Q, \P, \pmb{\zeta}$.

\end{itemize}

\begin{table}[t]
\caption{System parameters}
\centering
\resizebox{0.9\columnwidth}{!}{%
\begin{tabular}{l c} 
 \hline  \hline
 Simulation parameter (notation) & Value \\ [0.5ex] 
 \hline 
 Gaussian noise power spectral density ($\sigma^2_b, \sigma^2_u$) & $-196$ dBm/Hz \\
Operating frequency ($f$) & $0.8$ THz  \\
Bandwidth ($B$) & $10$ GHz\\
Reference channel power gain ($\beta_0$) &  $-71$ dB\\
Molecular absorption coefficient ($a_f$) & 0.005\\
UEs' average transmission power ($p^{ave}_k, \forall k$) & $20$ dBm\\
UEs' peak transmission power ($p^{max}_k \treq 4 p^{ave}_k, \forall k$) & $0.4$ W\\
UUR's average relaying power  ($p^{ave}_u$) & $0.4$ W\\
UUR's peak relaying power ($p^{max}_u \treq 4 p^{ave}_u$) & $1.6$ W\\
BS's average jamming power  ($p^{ave}_b$) & $0.5$ W\\
BS's peak jamming power ($p^{max}_b \treq 4 p^{ave}_b$) & $2$ W\\
UAV's operational altitude ($H$) & $10$ m\\
UAV's initial/final 2D location per flight ($\q_I$) & [$25$ m , $0$ m]\\
BS's horizontal location ($\q_b$) & [$0$ m, $0$ m]\\
Inner and outer radii of the region ($R_1, R_2$) & ($20$ m, $30$ m)\\
Number of randomly distributed UEs ($K$)& 5\\
Average flight power consumption budget ($P_{lim}$) & $200$ W\\
UAV's maximum velocity ($v^{max}_u$)& $20$ m.s\textsuperscript{-1}\\
UAV's maximum acceleration ($a^{max}_u$) & $5$ m.s\textsuperscript{-2}\\
Blade angular velocity ($\omega$) & $300$ rad.s\textsuperscript{-1} \\
Rotor radius ($r$) & $0.4$ m\\
Air density ($\rho$) & $1.225$ kg.m\textsuperscript{-3}\\
Rotor solidity factor ($s$) & 0.05\\
Rotor  disk area ($A\treq \pi r^2$) & $0.503$ m\textsuperscript{2}\\
Average hovering induced rotor velocity ($\nu_0 \treq \sqrt{\frac{W_u}{2\rho A}}$) & $4.03$ m.s\textsuperscript{-1}\\
Fuselage drag ratio ($d_0$) & 0.6\\ 
Profile drag coefficient ($\delta$) & 0.012\\
Incremental correction coefficient of induced power $(k_i)$& 0.1\\
UAV's weight ($W_t$) & $20$ Kg.m.s\textsuperscript{-2}\\
Blade profile constant ($P_0\treq \frac{\delta}{8}\rho s A \omega^3 r^3$) & $   79.856$ W\\
Parasite component constant ($P_i \treq (1+k_i)W_t\sqrt{\frac{W_t}{2\rho A}}$) & $88.63$ W \\
Mission time ($T$) & $10$ s\\
Time slot duration ($\delta_t$) & $0.1$ s\\
Overall iterative algorithm's convergence tolerance  ($\epsilon_1$) & 10\textsuperscript{-3}\\
Convergence tolerance of Dinkelbach's programming  ($\epsilon_2$) & 10\textsuperscript{-4}
\\[1ex] 
 \hline  \hline
\end{tabular}}
\label{table:1}
\end{table}

\begin{figure}[t]
\centerline{\includegraphics[width= 0.8\columnwidth]{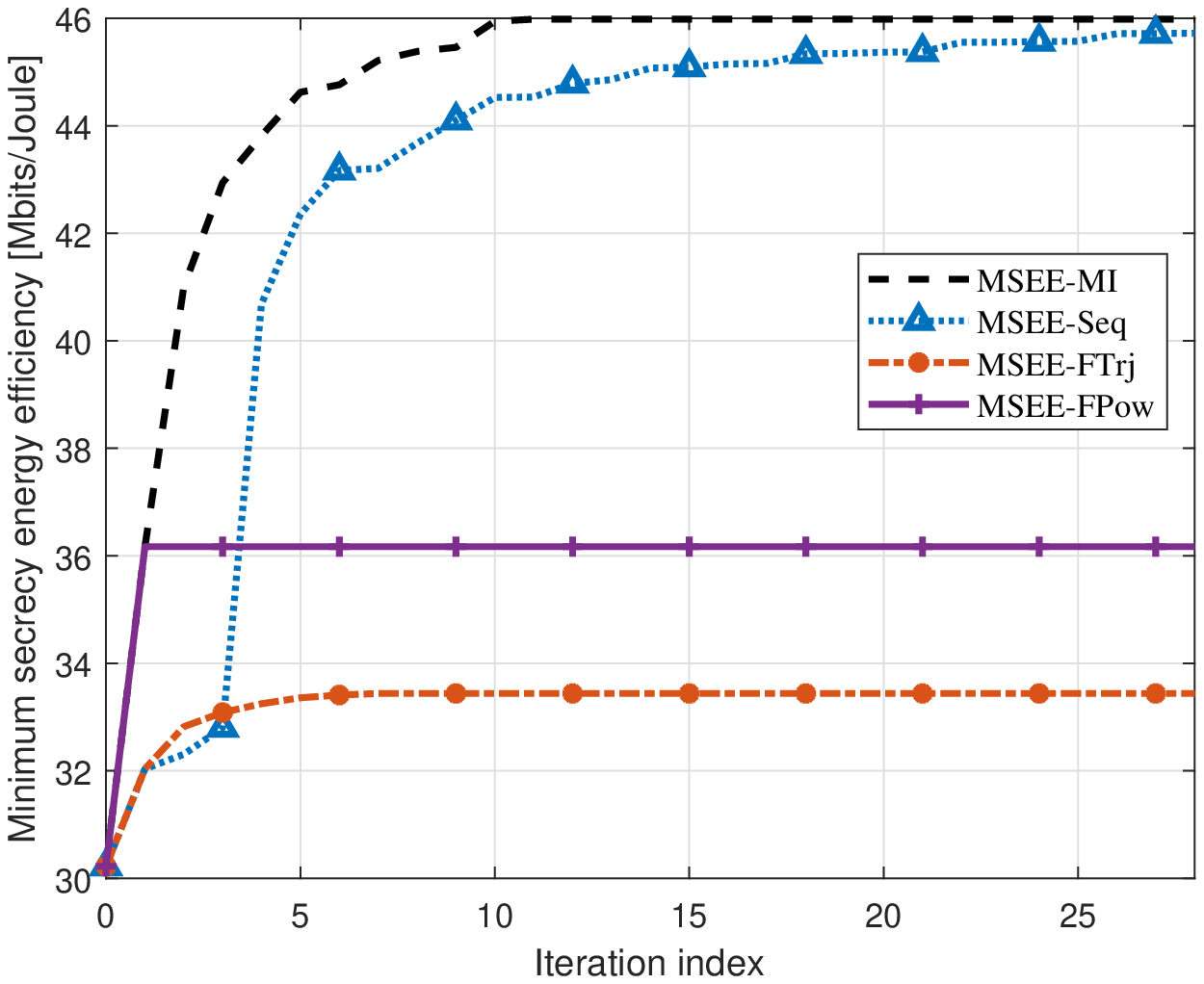}}
\caption{\textcolor{black}{Convergence verification of different MSEE optimization based algorithms.}}
\label{sim:fig1}
\end{figure}
Fig. \ref{sim:fig1} depicts the convergence of the proposed iterative algorithms for $T=100$ s. We can see that both benchmark schemes \textit{MSEE-FTrj} and \textit{MSEE-FPow} converge quickly; however, they both can achieve significantly lower MSEE performance than the proposed joint design of trajectory, power control and user scheduling schemes, i.e., \textit{MSEE-MI} and \textit{MSEE-Seq}. \textcolor{black}{Specifically, \textit{MSEE-MI} not only converges relatively faster than \textit{MSEE-Seq}, i.e., $12$ against $29$ iterations, but offers slightly higher MSEE than that of its counterpart, as well. However, they achieve approximately $52.1\%$ MSEE improvement, while  \textit{MSEE-FPow} and \textit{MSEE-FTrj} can only increase the MSEE by $19.7\%$ and $10.6\%$, respectively.}

\begin{figure}[t!]
 \centering
    \begin{subfigure}[t]{\columnwidth}
        \centering
        \includegraphics[width= 0.8\columnwidth]{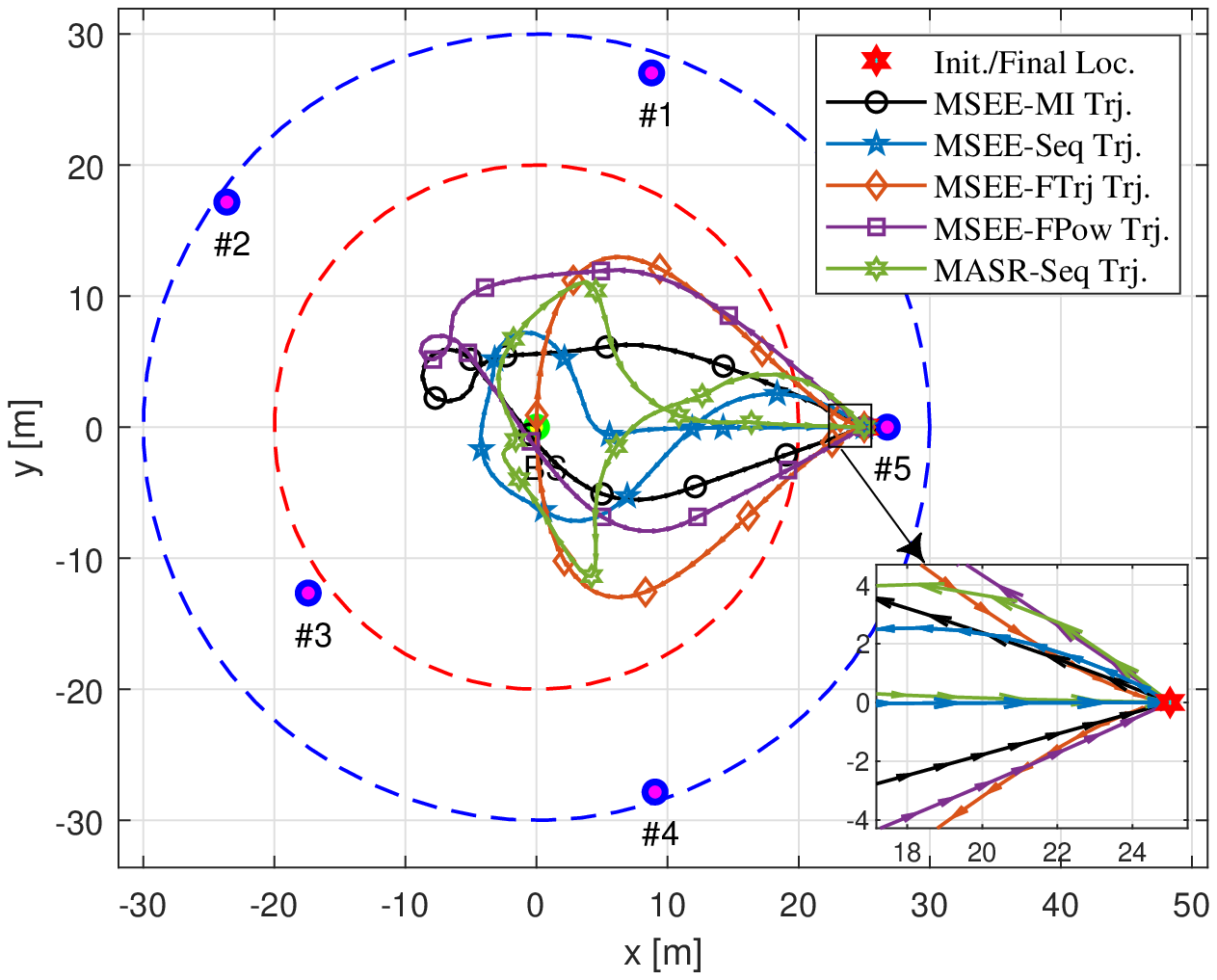}
        \caption{$T = 7$ s.}
        \label{sim:fig3_1}
    \end{subfigure} 
     
 \begin{subfigure}[t]{\columnwidth}
        \centering
        \includegraphics[width= 0.8\columnwidth]{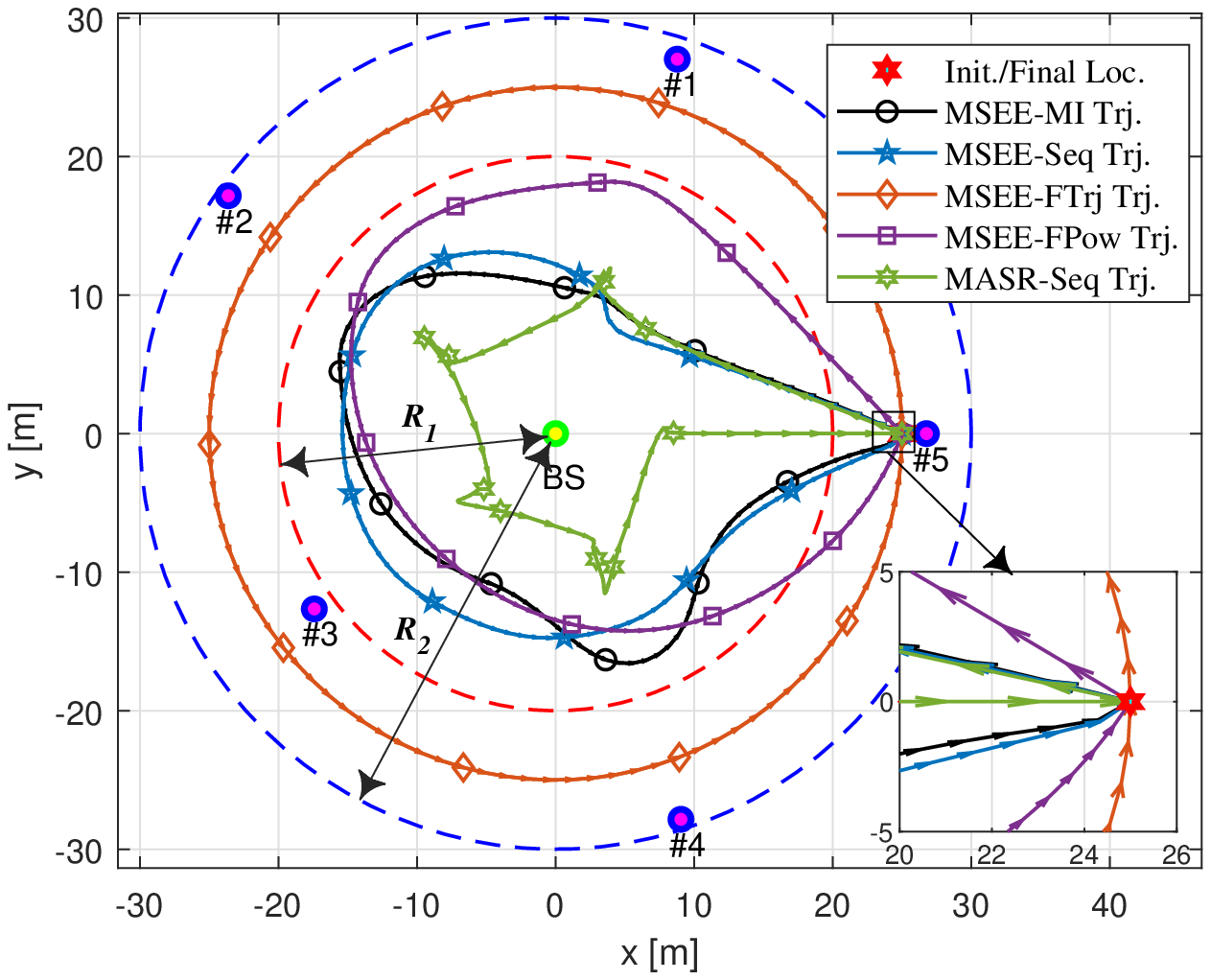}
        \caption{$T = 10$ s.}
        \label{sim:fig3_2}
\end{subfigure}
    
 \begin{subfigure}[t]{\columnwidth}
        \centering
        \includegraphics[width= 0.8\columnwidth]{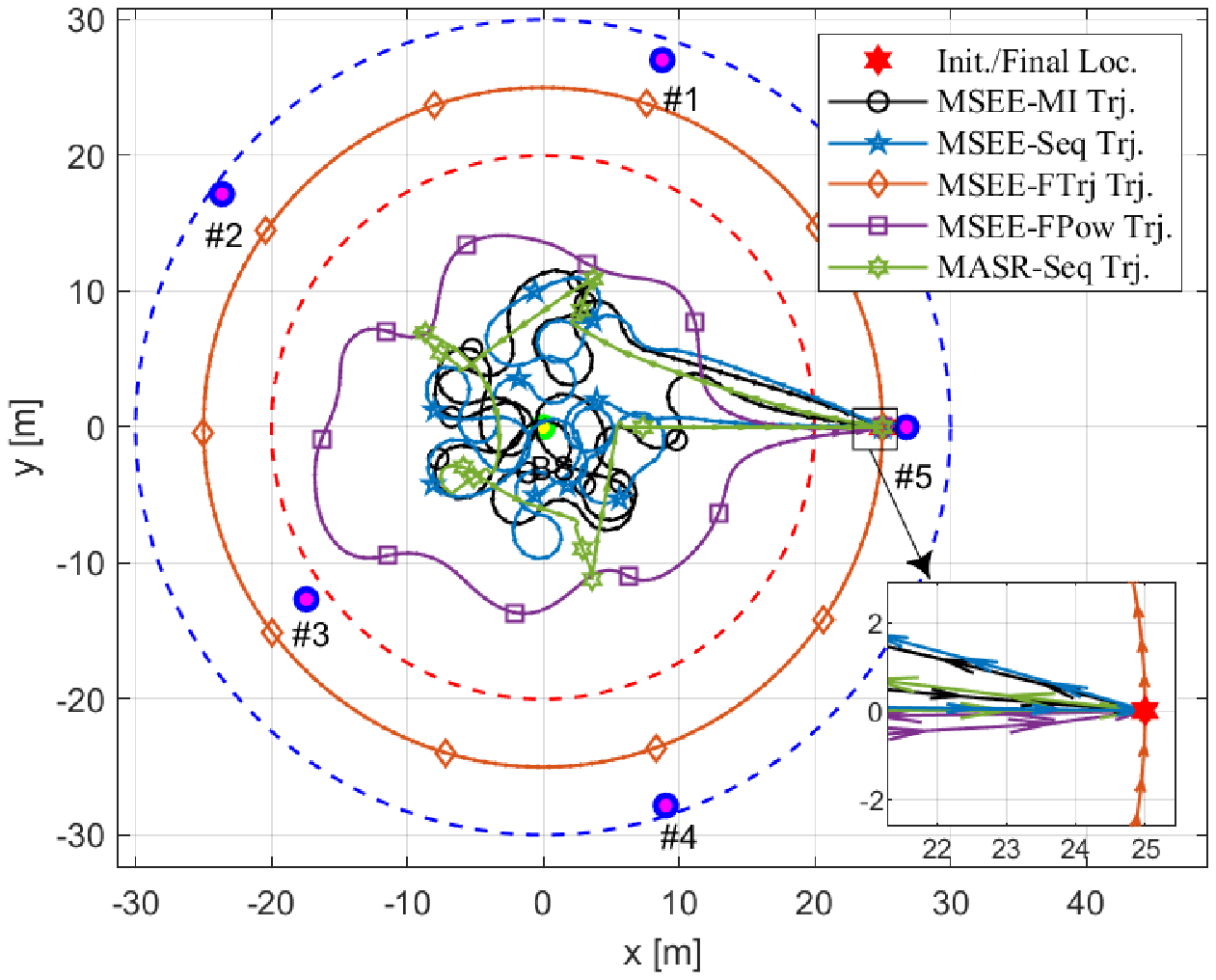}
        \caption{$T = 20$ s.}
        \label{sim:fig3_3}
\end{subfigure}
 \caption{\textcolor{black}{\label{sim:fig3} Designed trajectory based on different optimization algorithms and mission time $T$.}} 
\end{figure}


Fig. \ref{sim:fig3} illustrates UUR’s trajectories according to different optimization schemes with different mission duration $T=\{7, 10, 20\}$ s. 
We note that when the initial circular trajectory is impossible due to significantly low mission time, e.g., $T=7$ s, and owing to the UAV’s physical system limitations, the crucial task of path-planning can be efficiently designed based on a Piriform trajectory initialization as shown in Fig. \ref{sim:fig3_1}. However, when $T$ is sufficiently high, e.g., $T=10$ s or $T=20$ s, the baseline trajectory of circular shape can be utilized.
It should be mentioned that the curve belonging to the \textit{MSEE-FTrj} does represent the initial feasible cyclic trajectory based on the circular or Piriform shapes, and the other curves illustrate the optimized UUR’s trajectory according to the different algorithms. It can be observed that the optimized trajectories are much more complicated than the initial ones, particularly when the UUR enjoys relatively higher mission time and flexibility, i.e., $T=20$ s, as shown in Fig. \ref{sim:fig3_3}.  Notice that UUR should fly towards UEs’ locations to obtain data with low power. This, in turn, can significantly increase the chance of information leakage due to a stronger wiretap link and less effective BS’s jamming. Thus, UUR prefers to stay not too far from the BS. Overall, we see that the path planning makes UUR adjust trajectory through the best possible path, efficiently forming the distances between the UUR, selected UEs, and the BS such that a balanced trade-off between the channel conditions for the friendly jamming transmission in the first phase as well as the aerial relaying in the second phase of transmission improves the overall energy-efficient secrecy performance. Further, we observe from Fig. \ref{sim:fig3_2} that the MSEE-based trajectories are generally smoother than that of the \textit{MASR-Seq}, wherein the UUR prefers to quickly reach the best locations providing service for the designated UEs while hovering. This observation implies that the MSEE optimization demands
such smooth paths  for a lower flight power consumption of the UUR, in contrast to the \textit{MASR-Seq} design where the UUR’s velocity might harshly fluctuate for the minimum ASR (MASR) improvement if required.  
\begin{figure}[t]
\centerline{\includegraphics[width= 0.8\columnwidth]{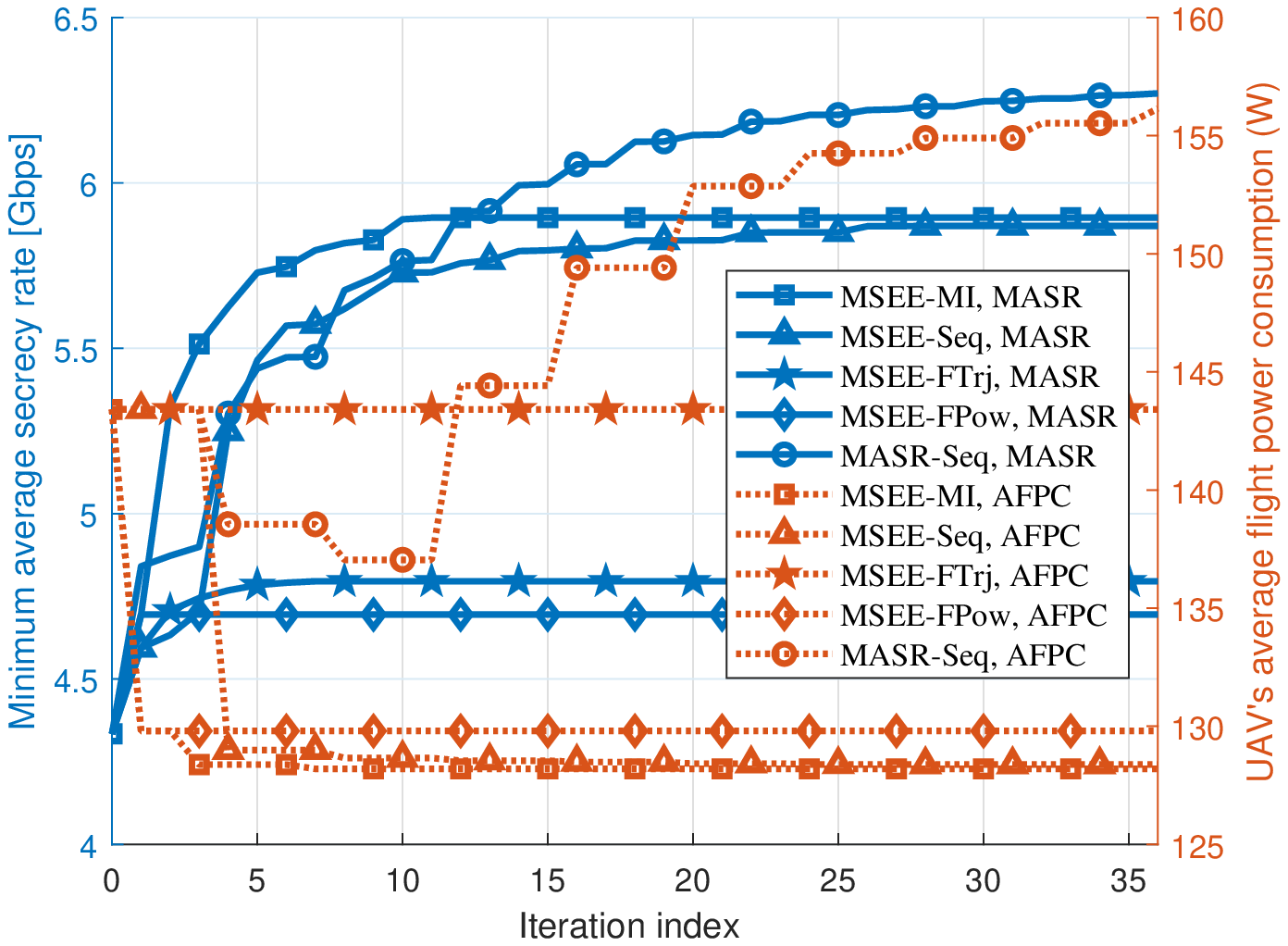}}
\caption{\textcolor{black}{Comparison between MASR and AFPC against iteration index for different algorithms.}}
\label{sim:fig4}
\end{figure}

Fig. \ref{sim:fig4} illustrates the MASR and the average flight power consumption (AFPC) against iteration indices for different schemes.
It is crystal clear that for the MSEE-based algorithms, the MASR and the AFPC performances tend to be non-decreasing and non-increasing, respectively. In contrast, for \textit{MASR-Seq} scheme, the AFPC first decreases down to some level then increases until convergence in $36$ iterations. We also note that this scheme can achieve slightly higher MASR performance than our proposed schemes but at the cost of significantly higher AFPC, resulting in lower MSEE ($43.13$ Mbits/Joule). This reinforces the significance of energy-efficient secure UAV-system design, which conventional works have somewhat ignored.
\begin{figure}[t]
\centerline{\includegraphics[width= 0.8\columnwidth]{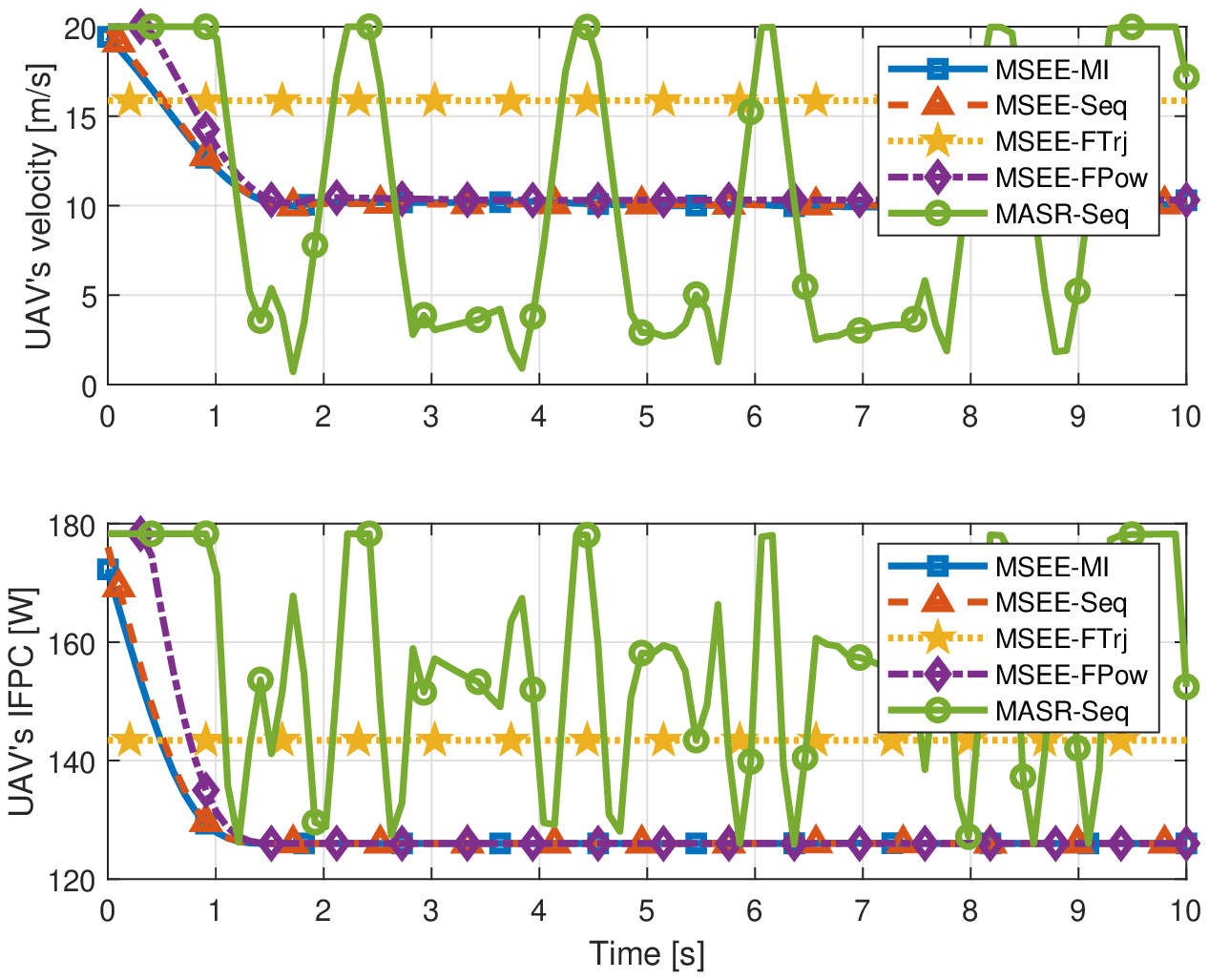}}
\caption{\textcolor{black}{UAV's instantaneous velocity and flight power consumption profile according to different schemes.}}
\label{sim:fig5}
\end{figure}

Fig. \ref{sim:fig5} is plotted to demonstrate how the UAV's velocity (Vel.) and the instantaneous flight power consumption (IFPC) are adjusted over time based on different algorithms. The curves labeled with \textit{"MSEE-FTrj"} basically represent the initial UAV's velocity and IFPC for all the other scenarios. We observe that all the trajectory optimization schemes make UAV fly with roughly less speed variation for a relatively more extended period of time (e.g., from $2$s to $10$s) to satisfy mission requirements as well as improve the MSEE performance. Specifically, UUR starts at a high initial speed to fast reach the targeted location while gradually decreasing the speed down to some appropriate level, maintaining comparatively unchanged afterward for the sake of efficient power consumption purposes. Nonetheless, based on the \textit{MASR-Seq} scheme, the velocity tends to be minimal while confidential relaying towards the scheduled UE; however, it changes drastically, enabling the mission to be accomplished by the end of the specified time. As it can be seen, this approach consumes relatively high flight power.

\begin{figure*}[t!]
\centering
\hspace{12mm}\begin{subfigure}[t]{0.3\textwidth}
  \centering
  \includegraphics[width= \textwidth]{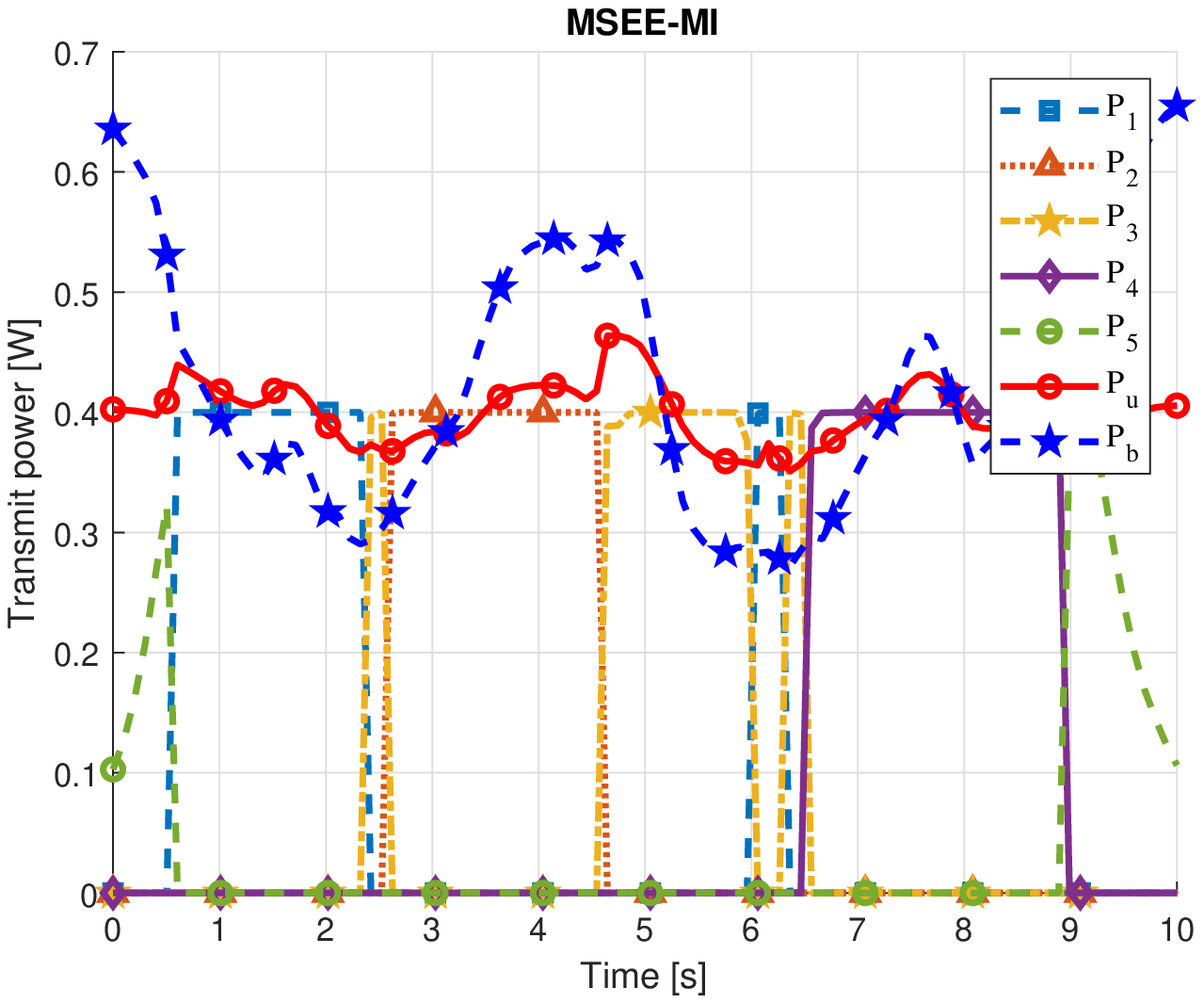}
  \caption{}
  \label{fig:sfig1}
\end{subfigure}
\begin{subfigure}[t]{0.3\textwidth}
  \centering
  \includegraphics[width= \textwidth]{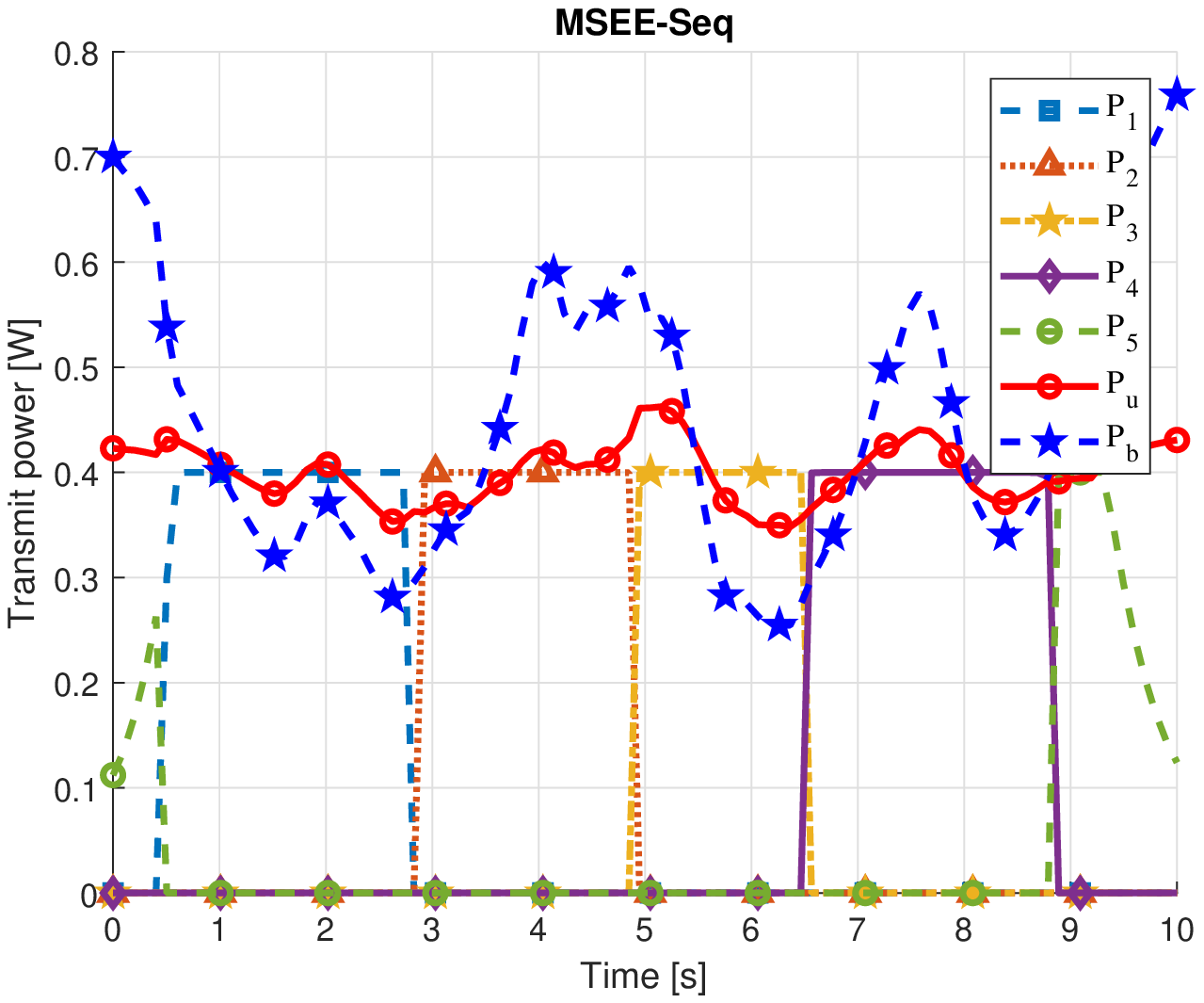}
  \caption{}
  \label{fig:sfig2}
\end{subfigure}
\newline
\begin{subfigure}[t]{0.3\textwidth}
  \centering
  \includegraphics[width=\textwidth]{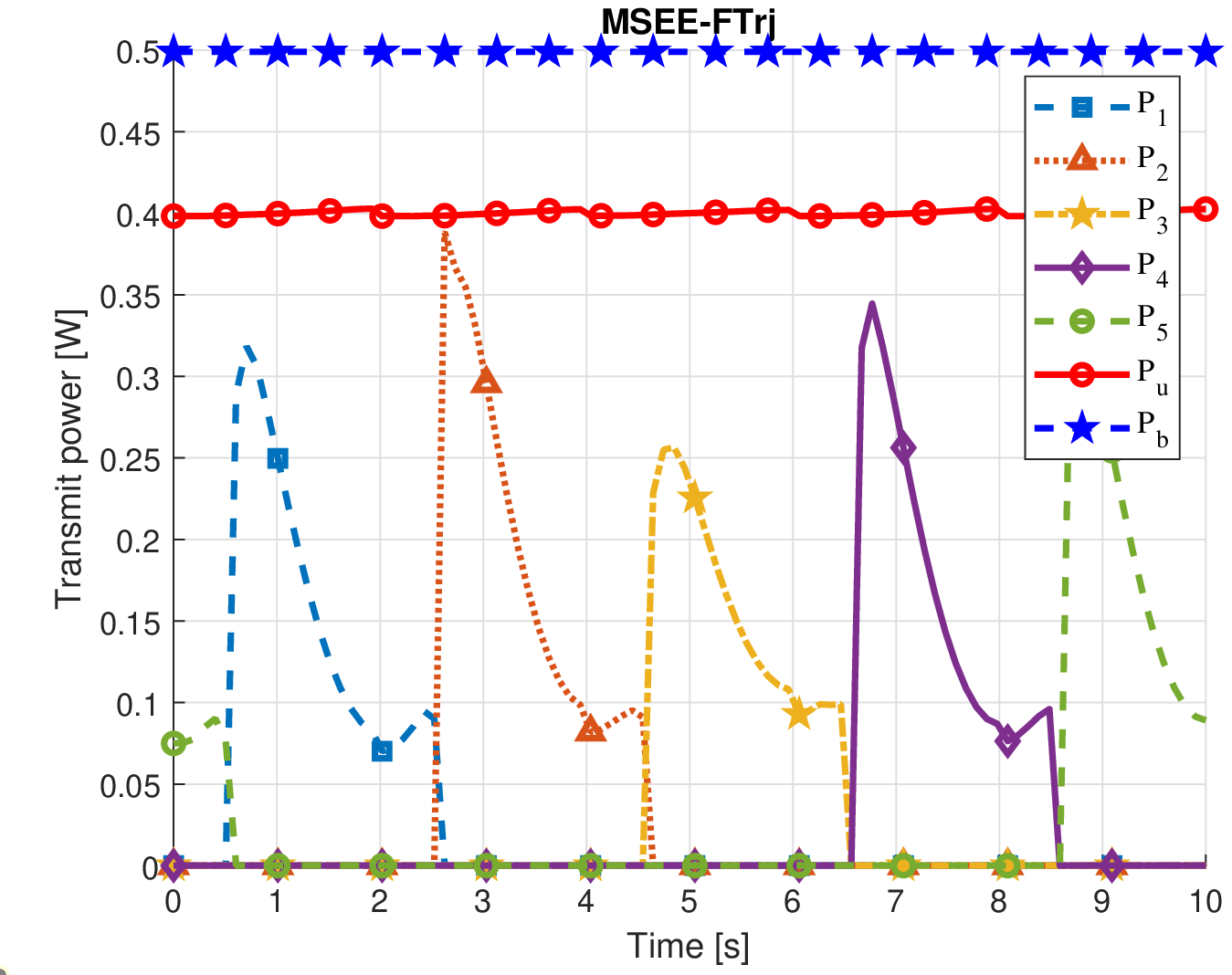}
  \caption{}
  \label{fig:sfig3}
\end{subfigure}
\begin{subfigure}[t]{0.3\textwidth}
  \centering
  \includegraphics[width= \textwidth]{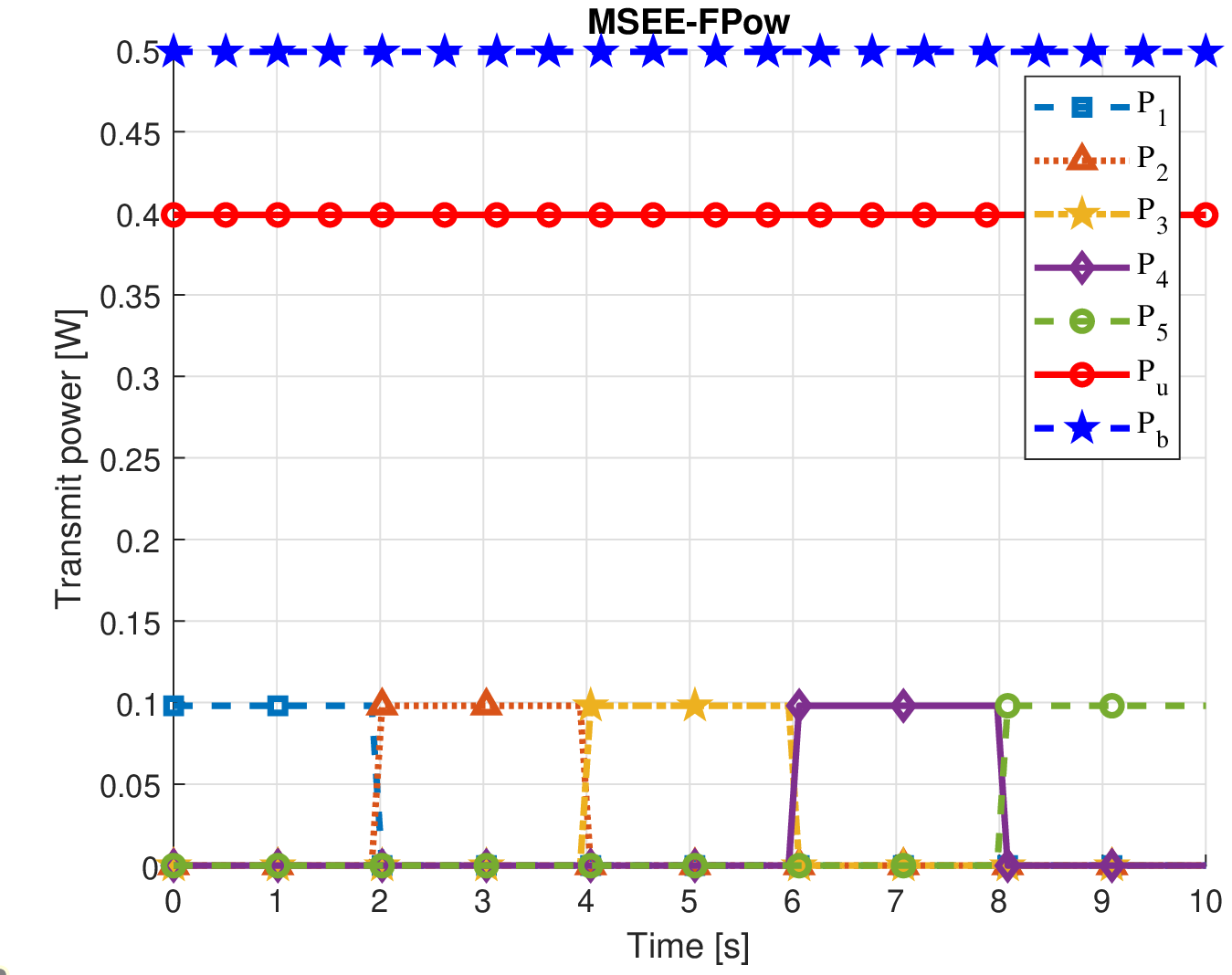}
  \caption{}
  \label{fig:sfig4}
\end{subfigure}
\begin{subfigure}[t]{0.3\textwidth}
  \centering
  \includegraphics[width=\textwidth]{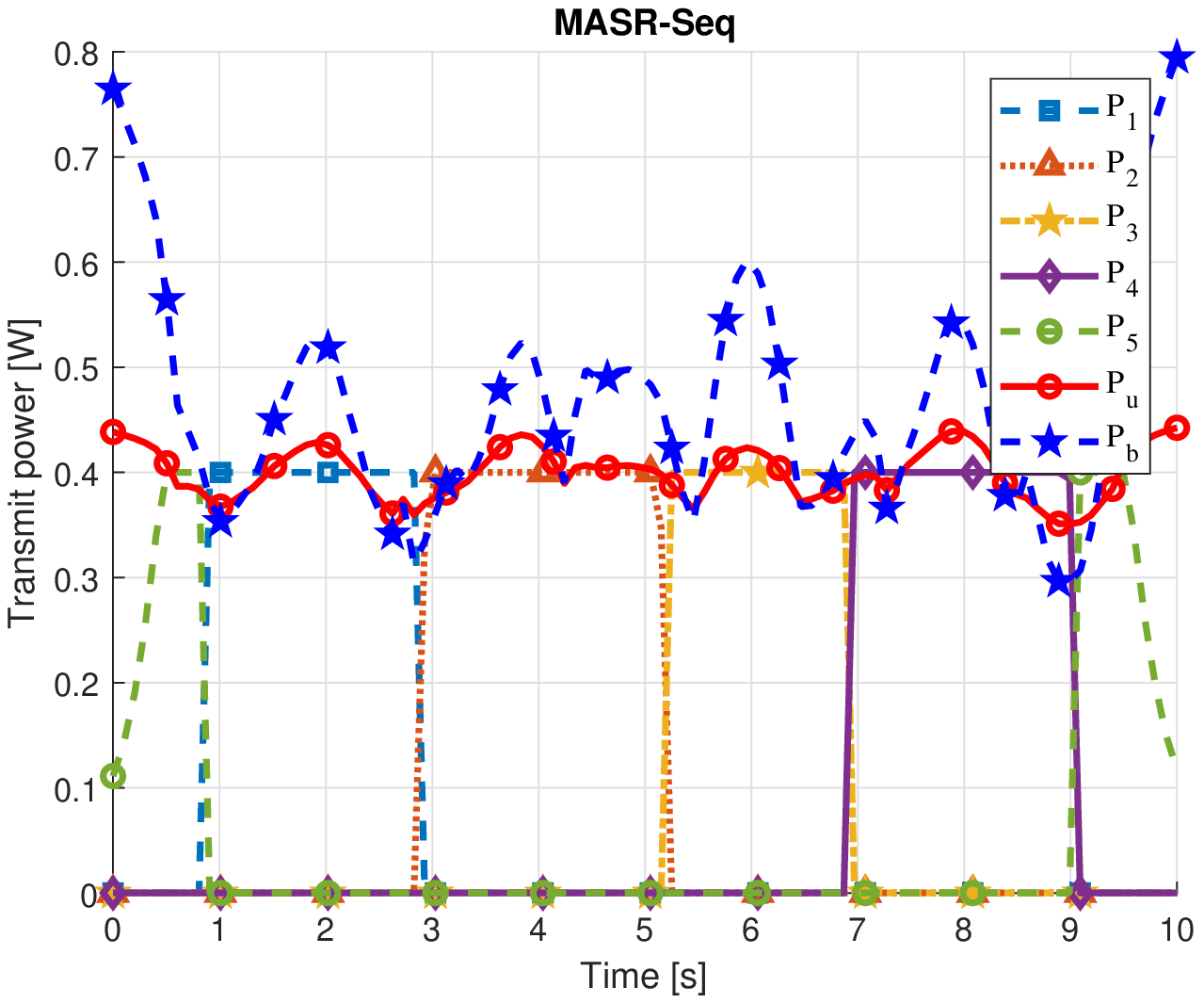}
  \caption{}
  \label{fig:sfig5}
\end{subfigure}%
\caption{\textcolor{black}{Transmit power allocation and user scheduling against time for different schemes.}}
\label{fig:fig6}
\end{figure*}

Fig. \ref{fig:fig6} illustrates the joint power allocation and user scheduling vs. time for different MSEE optimization algorithms. The sub-figure \ref{fig:sfig4} represents the non-optimal but feasible power allocations and user scheduling adopted for initialization of all the schemes. Initially, UUR is very close to $\mathrm{UE}_{5}$ but far from the BS. Hence, $\mathrm{UE}_{5}$ is scheduled due to a possibly better channel condition than the others, and the BS jams in high power while $\mathrm{UE}_{5}$ keeps low power.  For \textit{MSEE-FTrj}, UUR follows the circular trajectory while maintaining the same distance from the BS that has a constant jamming power. In contrast, sub-figures \ref{fig:sfig1}, \ref{fig:sfig2}, and \ref{fig:sfig5} show that at initial stage, $\mathrm{UE}_{5}$ increases power when UUR heads towards the BS and the BS reduces jamming power. Further, these UEs are scheduled unequally, but during their scheduling, except $\mathrm{UE}_5$, they need to utilize their maximum transmission powers for sending information, and the relaying power slightly fluctuates around $p^{ave}_u$.



\begin{figure}[t]
\centerline{\includegraphics[width= 0.8\columnwidth]{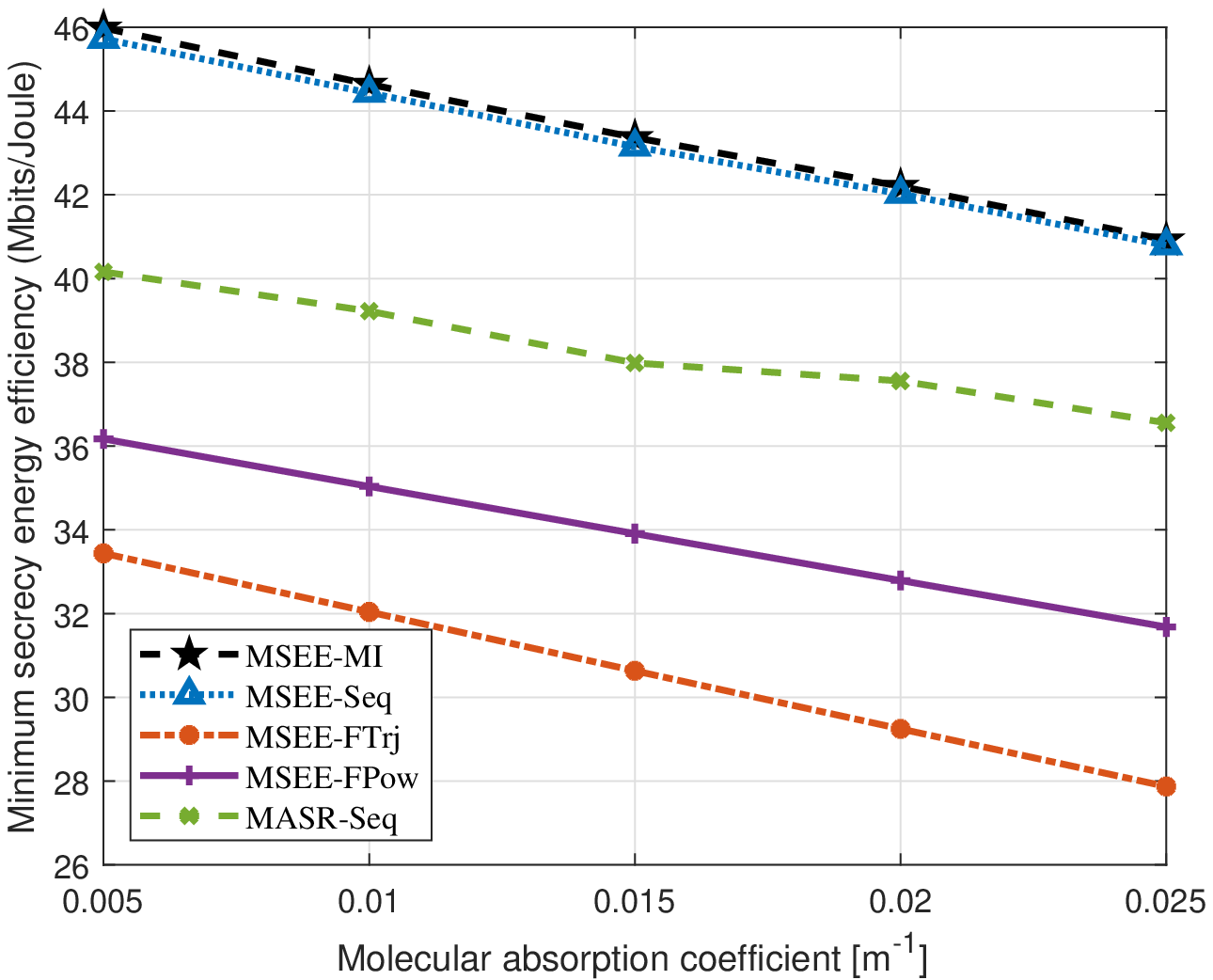}}
\caption{\textcolor{black}{Effect of THz molecular absorption on the MSEE.}}
\label{sim:fig7}
\end{figure}
\textcolor{black}{Fig. \ref{sim:fig7} depicts how the MSEE performance varies when the molecular absorption coefficient of THz links changes from $a_f=0.005$ to $a_f=0.025$ (similar range is also adopted in \cite{xu2021joint}), which can be physically translated to different carrier frequencies and environmental effects.  From the figure, we can see that the larger the molecular absorption coefficient, the lower the MSEE performance for all the schemes due to higher propagation loss arising from severe molecular absorption.  It is worth pointing out that the increased propagation loss results in the reduction of not only UUR's information leakage, but also BS's reception quality. Nonetheless, our proposed \textit{MSEE-MI} and \textit{MSEE-Seq} schemes substantially outperform the others regardless of the environmental conditions and operational carrier frequency of THz links. }

\begin{figure}[t]
\centerline{\includegraphics[width= 0.8\columnwidth]{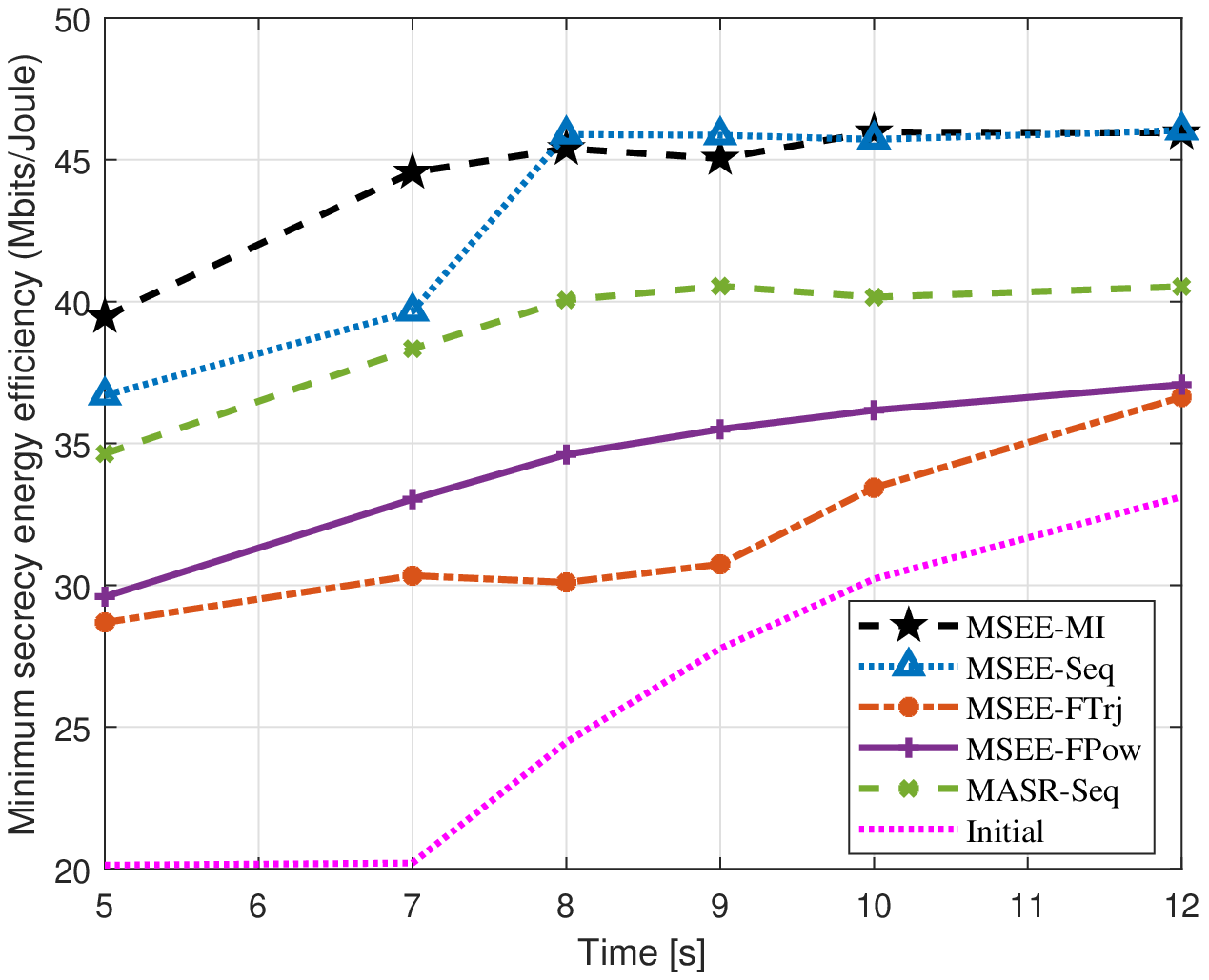}}
\caption{\textcolor{black}{Effect of mission time on the MSEE performance.}}
\label{sim:fig8}
\end{figure}
\textcolor{black}{In Fig. \ref{sim:fig8}, we investigate how the mission time $T$ impacts the MSEE performance for different schemes. Evidently, for the fixed trajectory schemes, i.e., \textit{MSEE-FTj} and the non-optimal initial feasible scheme, labeled as \textit{Initial}, as $T$ increases, MSEE also gets increased. however, such monotonically increasing trend is not observed on the other schemes. Indeed, when the mission time increases, the MSEE performance improves due to more time for secure communications and adjusting flight parameters. However, the higher the mission time, the larger the mechanical power consummation. Therefore, the overall trade-ff between these two phenomena, i.e., MASR and AFPC, results in the fact that the MSEE performance does not get monotonically increased, though following an overally increasing trend, as $T$ rises, e.g.,  from $8$ to $9$ s of \textit{MSEE-MI} and \textit{MSEE-Seq} curves,   in contrast to the MASR metric studied in \cite{tatarmamaghani2021joint}. This illustrates the significance of  considering propulsion power consumption for designing secure energy-efficient UAV-enabled systems. We also note that for a particular MSEE requirement, minimizing the task completion time of the considered UUR-system appears an interesting problem and requires deep investigation.}



\begin{figure}[t]
\centerline{\includegraphics[width= 0.8\columnwidth]{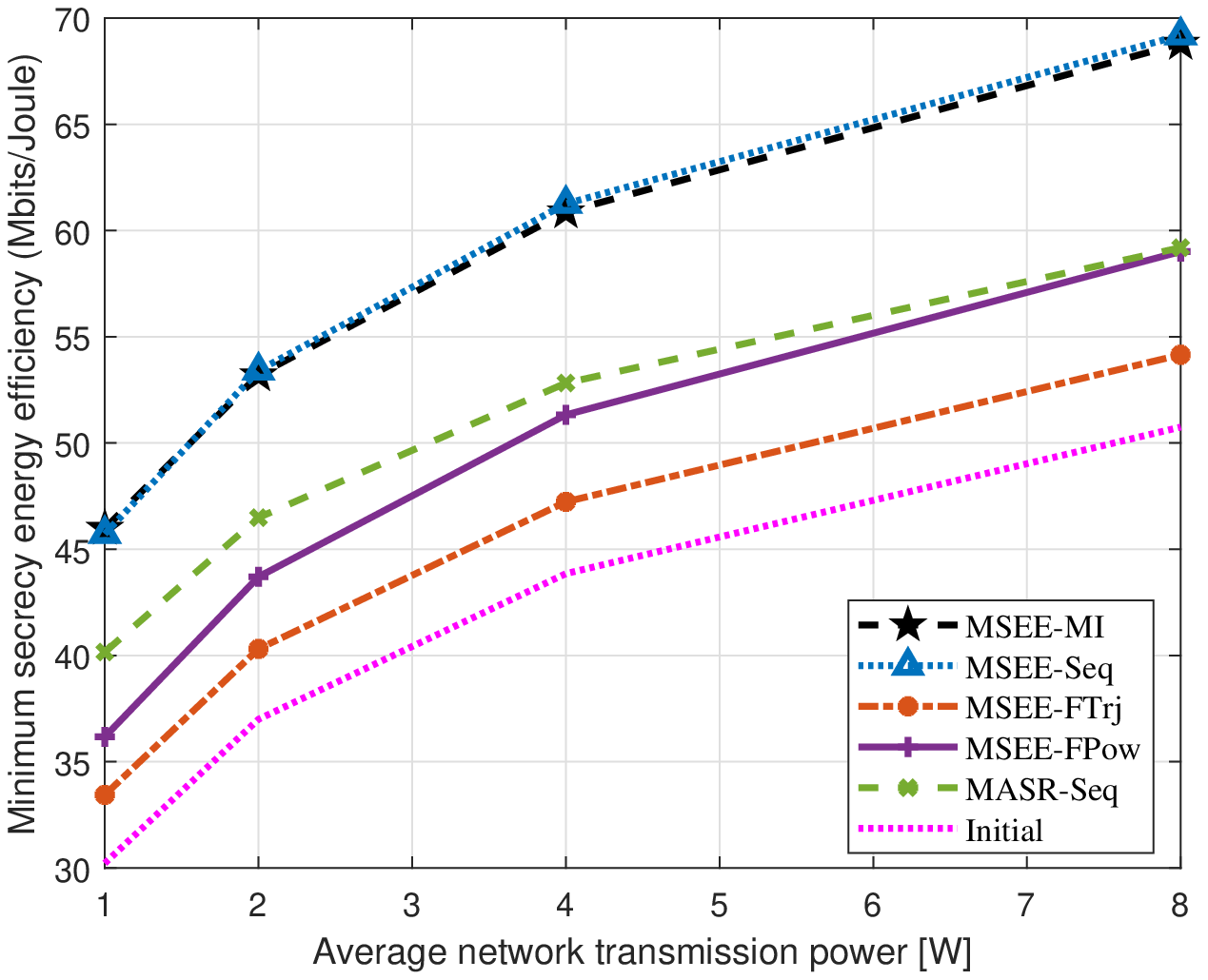}}
\caption{\textcolor{black}{Effect of average network transmission power on the MSEE performance.}}
\label{sim:fig9}
\end{figure}

\textcolor{black}{Fig. \ref{sim:fig9} illustrates the performance of MSEE with the increase of average network transmission power (ANTP) parameter, defined as $p^{ave}$ such that $p^{ave}_k = 0.1 p^{ave}$, $p^{ave}_b = 0.5 p^{ave}$, $p^{ave}_u = 0.4 p^{ave}$. The curve labeled as \textit{Initial} belongs to the MSEE performance of non-optimal feasible initialization. We again observe that the performance gaps that our joint designs offer are significantly higher than the other benchmark designs, and the relative gap, interestingly, slightly increases as the ANTP gets larger. For example, the \textit{MSEE-MI} scheme improves the MSEE performance by approximately $16$ Mbits/Joule when $p^{ave}=1$ W; however, around $19$  Mbits/Joule enhancement is achieved with $p^{ave}=8$ W, in  comparison with the \textit{Initial}  scheme. }

\begin{figure}[t]
\centerline{\includegraphics[width= 0.8\columnwidth]{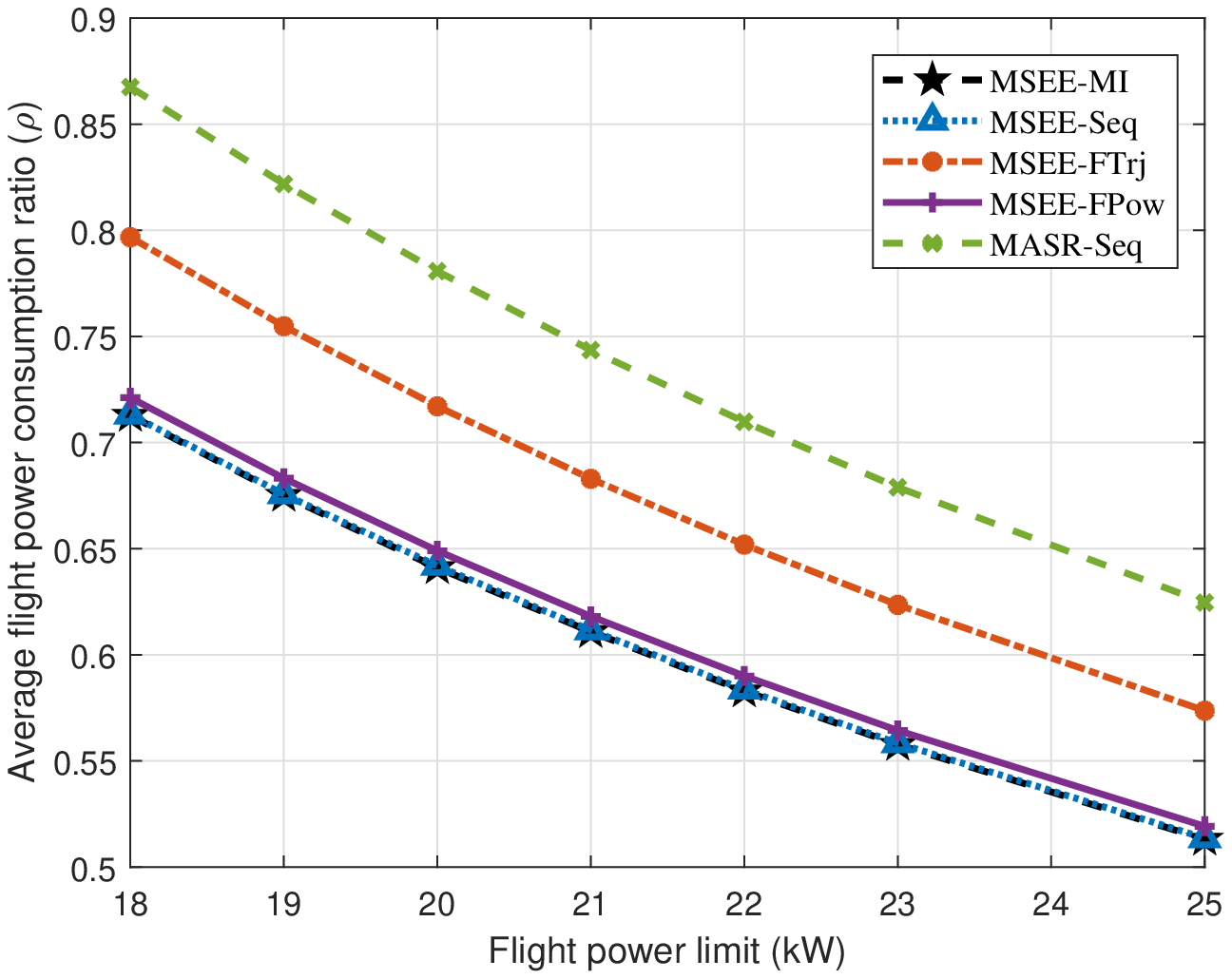}}
\caption{\textcolor{black}{Effect of flight power limit on the average flight power consumption ratio.}}
\label{sim:fig10}
\end{figure}
\textcolor{black}{In Fig. \ref{sim:fig10}, the average flight power consumption ratio (AFPCR), defined as $\rho \treq \frac{\bar{\mathbf{P}}_f}{\bar{P}_{lim}}$ is plotted against the flight power limit. It can be seen that the MSEE designs are more conservative than the \textit{MASR-Seq} scheme when it comes to the AFPCR performance, particularly, both \textit{MSEE-MI} and \textit{MSEE-Seq} achieve the least ANTPR, comparatively. Further, we can construe that as the flight power limit decreases, the larger proportion of the dedicated UAV's on-board battery resource is consumed  throughout the mission.}

\section{Conclusions}\label{sec:conclusion}
In this paper, we investigated the challenging task of designing
an energy-efficient THz-UUR system for secure and periodically data delivering from multiple ground UEs towards the BS. For the fairness of QoS amongst the UEs, a MSEE maximization problem was formulated, by which the fundamental system parameters are designed to improve the overall system secrecy and energy-efficiency performance. This was formally posed as a challenging mixed-integer non-convex nonlinear maximin optimization problem.  We then embarked on tackling the non-convexity of the formulated problem and proposed low-complex  BCD-SCA-Dinkelbach based iterative algorithms to solve it suboptimally with guaranteed convergence. Simulation results confirmed the fast convergence of our proposed algorithms, demonstrated significant MSEE performance improvement than the other benchmarks, and provided insightful results in the optimized system parameters such as UUR’s trajectory and velocity pattern as well as communication resource allocations, including transmit power profiles and user scheduling. Also, the effects of mission time, and molecular absorption factors arising from the THz links on the system MSEE performance have been examined. \textcolor{black}{As future work, we will deeply investigate the dynamic topology of aerial platforms with more practical THz channel modeling while leveraging benefits of extreme directional beamforming for intelligent UUR systems as well as taking into account the mobility of terrestrial UE}.


\appendices
\numberwithin{equation}{section}
\makeatletter 
\newcommand{\section@cntformat}{Appendix \thesection:\ }
\makeatother

\section{Proof of Lemma \ref{lemma1}}\label{Appendix A}

Computing the Hessian matrices of given functions yields
\begin{align}
    \mathbf{H}_1 &=\nabla^2(Z_1) = \frac{c^2}{(x+cy)^2}\begin{bmatrix}
-\frac{y^2}{x} & y\\
y & -x
\end{bmatrix},\\
\mathbf{H}_2 &=\nabla^2(Z_2) =\begin{bmatrix}
-\frac{aby^2\sigma_3}{\sigma_2} & \sigma_1\\
\sigma_1 & -\frac{abx^2\sigma_3}{\sigma_2}
\end{bmatrix}, 
\end{align}
wherein $\nabla^2(\cdot)$ denotes the hessian operator, and $\sigma_1 \treq \frac{abxy\sigma_3}{\sigma_2}$, $\sigma_2 \treq (y+bx)^2((a+1)y+bx)^2$,
and $\sigma_3 \treq 2(1+a)y + (a+2)bx$. One can verify that both matrices $\mathbf{H}_1$ and $\mathbf{H}_2$ are negative semidefinite ($\mathbf{H}_{1(2)} \prec 0$), thereby functions $Z_1$ and $Z_2$ are jointly concave w.r.t $x$ and $y$. Accordingly, there exist global over-estimators for concave functions in \eqref{Z2_approx}, following the first-order convexity condition law \cite{cvx_boyd}. The proof is completed.

\section{Proof of Lemma \ref{lemma3}}\label{Appendix B}
Computing the gradients of given functions w.r.t $x$ and $y$ yields
\begin{align}
\nabla f_{41}(x, y) &= \left[\begin{array}{c}
-\frac{a}{{\left(a\,x+b\,y\right)}\,{\left(a\,x+b\,y+1\right)}}\\
-\frac{b}{{\left(a\,x+b\,y\right)}\,{\left(a\,x+b\,y+1\right)}}
\end{array}\right],\\
\nabla f_{42}(x, y) &= \left[\begin{array}{c}
-\frac{c\,y}{x\,{\left(c\,y+d\,x+x\,y\right)}}\\
-\frac{d\,x}{y\,{\left(c\,y+d\,x+x\,y\right)}}
\end{array}\right],\\
\nabla f_{43}(x) &= x\,{\mathrm{e}}^{e\,x} \,{\left(e\,x+2\right)},
\end{align}
Further, calculating the Hessian matrix of the functions $f_{41}(x,y)$ and $f_{42}(x,y)$ and the second order derivative of $f_{43}(x)$, we can reach
\begin{align}
 \hspace{-3mm}  \mathbf{H}_{41}\hspace{-1mm}\treq\hspace{-1mm} \nabla^2{f_{41}}(x,y)&\hspace{-1mm}=\hspace{-1mm}
\begin{array}{l}
\hspace{-1mm}\frac{\left(2\,a\,x+2\,b\,y+1\right)}{{{\left(ax+by\right)^2\left(ax+by+1\right)^2}}}\hspace{-1mm}\left[\begin{array}{cc}
a^2 & ab\\
ab & b^2
\end{array}\right]\hspace{-1mm},
\end{array}\\
  \hspace{-5mm} \mathbf{H}_{42}\hspace{-1mm}\treq\hspace{-1mm} \nabla^2{f_{42}}(x,y)\hspace{-1mm}&= \hspace{-1mm}
\begin{array}{l}
\hspace{-2mm}\left[\hspace{-2mm}\begin{array}{cc}
\frac{c\,y\,{\left(c\,y+2\,d\,x+2\,x\,y\right)}}{x^2 \,{{\left(c\,y+d\,x+x\,y\right)}}^2 } & \hspace{-2mm}  -\frac{c\,d}{{{\left(c\,y+d\,x+x\,y\right)}}^2 }\\
-\frac{c\,d}{{{\left(c\,y+d\,x+x\,y\right)}}^2} & \hspace{-2mm} \frac{d\,x\,{\left(2\,c\,y+d\,x+2\,x\,y\right)}}{y^2 \,{{\left(c\,y+d\,x+x\,y\right)}}^2 }
\end{array}\hspace{-2mm}\right],\\
\end{array}\\
\nabla^2 f_{43}(x) &= {\mathrm{e}}^{e\,x} \,{\left(e^2 \,x^2 +4\,e\,x+2\right)} \geq 0, \label{f43_cvx}
\end{align}
We can verify that the first-order and second-order determinants of $\mathbf{H}_{41}$ and $\mathbf{H}_{42}$ are all non-negative, and therefore, the Hessian matrices are positive semi-definite ($\mathbf{H}_{41(2)}\succ 0$), indicating that functions  $f_{41}(x, y)$, $f_{42}(x, y)$ are jointly convex w.r.t $x$ and $y$. Further, the convexity of $f_{43}(x,r)$ and $f_{44}(x,p)$ follows from \eqref{f43_cvx} and \cite[Lemma 1]{mamaghani2020improving}. Given these functions are all convex, one can use the first-order Taylor expansions at points $x_0$ and $y_0$ to reach the global tight lower-bounds and inequalities in Lemma \ref{lemma2}. The proof is completed. 

\bibliographystyle{IEEEtran}
\bibliography{References}

\begin{IEEEbiography}[{\includegraphics[width=1in,height=1.25in,clip,keepaspectratio]{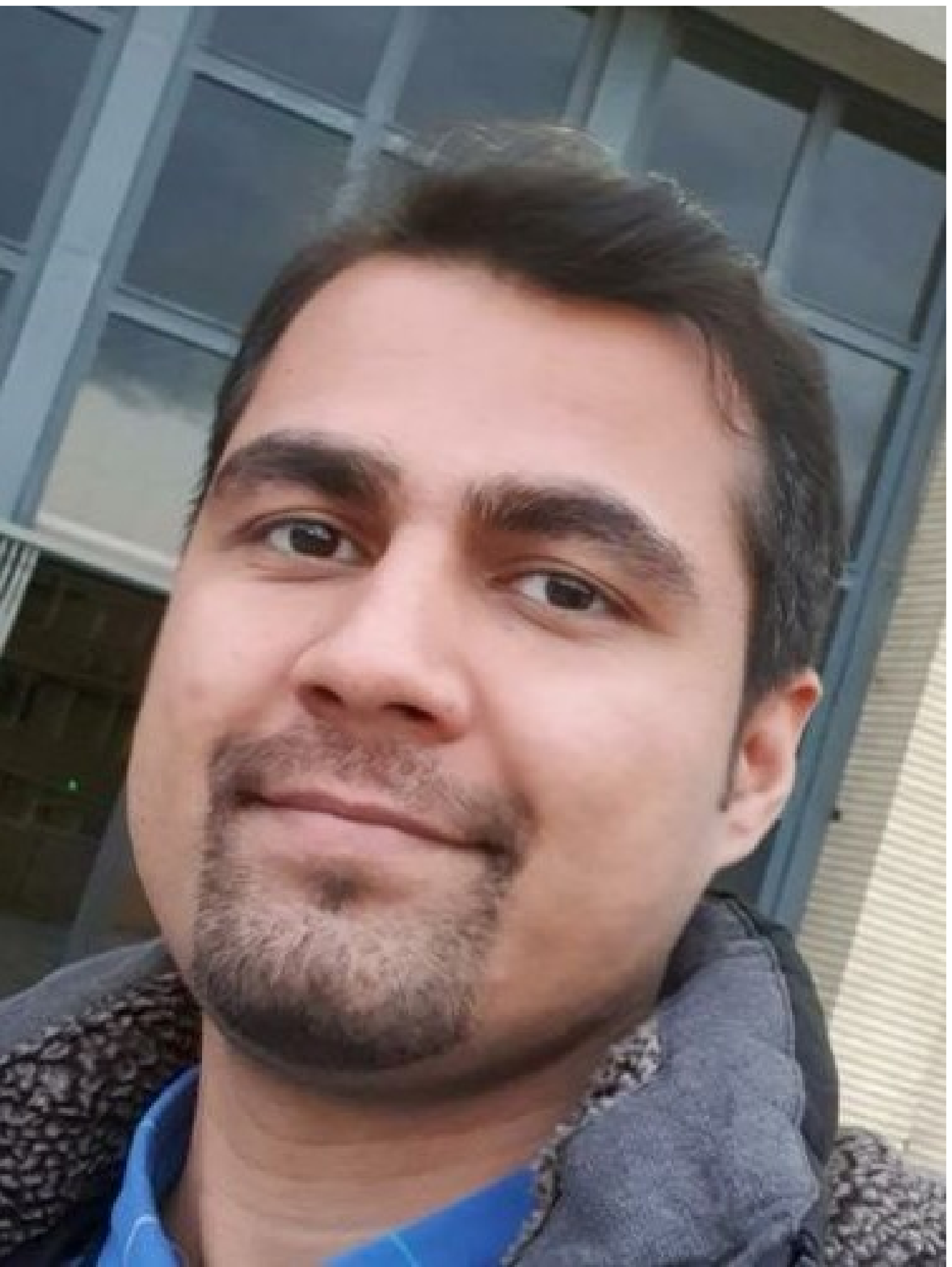}}]{Milad Tatar Mamaghani}(GS'20) was born in Tabriz, Iran, on May 12, 1994. He earned dual B.Sc. (Hons) degrees in electrical engineering fields - Telecommunications and Control - from the Amirkabir University of Technology, Tehran, Iran, in 2016 and 2018, respectively. He is currently pursuing the Ph.D. degree with the Department of Electrical and Computer Systems Engineering, Monash University, Melbourne, Australia. He is the author of several papers published in prestigious journals, and has served as a volunteer reviewer of various reputable publication venues such as TWC, TIFS, TVT, TCOM, TCCN, TMC, ISJ, Access, WCL, etc. His research interests mainly focus on beyond 5G wireless communications and networking, physical-layer security, UAV communications, optimization, and artificial intelligence. He is a member of the IEEE Communications Society and the IEEE Signal Processing Society.
\end{IEEEbiography}

 \vskip 0pt plus -1fil

\begin{IEEEbiography}[{\includegraphics[width=1in,height=1.25in,clip,keepaspectratio]{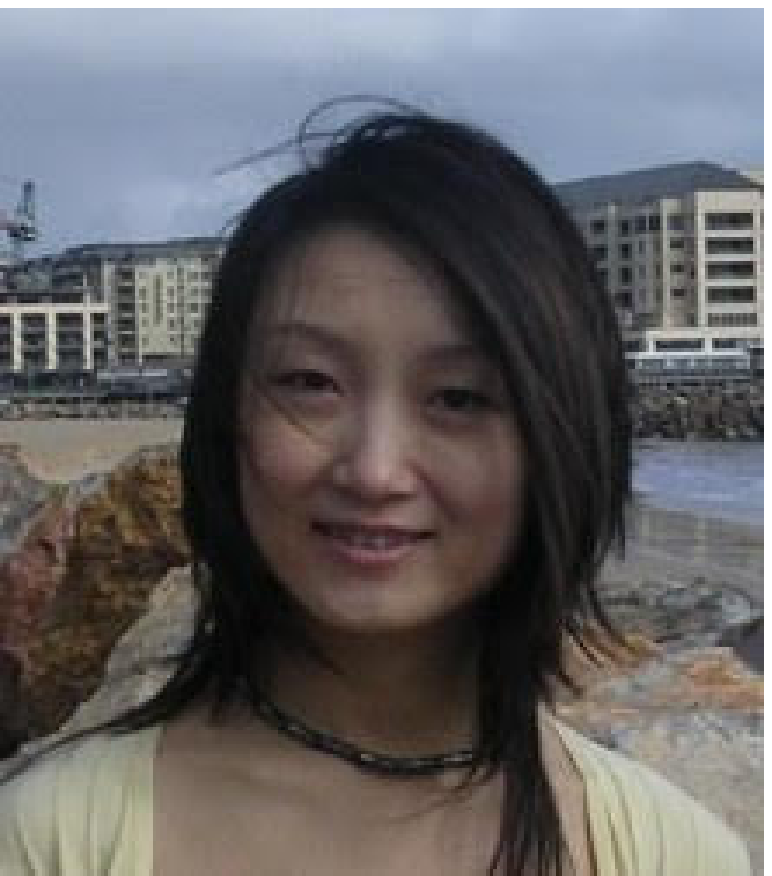}}]%
{Yi Hong}(S'00--M'05--SM'10)
is currently an Associate Professor at the Department of Electrical and Computer Systems Eng.,
Monash University, Melbourne, Australia.
She obtained her Ph.D. degree in Electrical Engineering and Telecommunications 
from the University of New South Wales (UNSW), Sydney, and received   
the {\em NICTA-ACoRN Earlier Career Researcher Award} at the 2007 {\em Australian Communication
Theory Workshop}, Adelaide, Australia. She served on the Australian Research Council College of Experts (2018-2020). Yi Hong is currently an Associate Editor for {\em IEEE Transactions on Green Communications and Networking}, and was the Associate Editor for {\em IEEE Wireless Communications Letters} and {\em Transactions on Emerging Telecommunications Technologies (ETT)}.
She was the General Co-Chair of {\em IEEE Information Theory Workshop} 2014, Hobart; the Technical Program Committee Chair of
{\em Australian Communications Theory Workshop} 2011, Melbourne; and the Publicity Chair
at the {\em IEEE Information Theory Workshop} 2009, Sicily. She was a Technical Program Committee member for
many IEEE leading conferences. Her research interests include
communication theory, coding and information theory with applications to telecommunication engineering.
\end{IEEEbiography}

\end{document}